\documentclass[journal,final,letterpaper,twoside,twocolumn]{IEEEtran}

\makeatletter
\def\bstctlcite{\@ifnextchar[{\@bstctlcite}{\@bstctlcite[@auxout]}}
\def\@bstctlcite[#1]#2{\@bsphack
  \@for\@citeb:=#2\do{%
    \edef\@citeb{\expandafter\@firstofone\@citeb}%
    \if@filesw\immediate\write\csname #1\endcsname{\string\citation{\@citeb}}\fi}%
  \@esphack}
\makeatother

\usepackage{amsmath}
\usepackage{amsthm}
\usepackage{amssymb}
\usepackage{mathrsfs} 
\newcommand{\si}[1]{#1}
\newcommand{\second}{$\mathrm{s}$}
\newcommand{\nano}{$\mathrm{n}$}
\newcommand{\metre}{$\mathrm{m}$}
\newcommand{\SIrange}[3]{$#1$#3 to $#2$#3}

\usepackage{dsfont} 
\usepackage{mathtools}
\usepackage{stmaryrd}  
\usepackage{enumitem} 

\usepackage[ruled,vlined]{algorithm2e} 

\usepackage{booktabs}
\usepackage{multirow} 
\usepackage[figuresleft]{rotating} 

\usepackage[noadjust]{cite}
\usepackage{graphicx}
\DeclareGraphicsExtensions{.pdf,.jpg,.png}
\usepackage{textcomp} 

\usepackage{url}
\usepackage{hyperref}  
\hypersetup{colorlinks,%
            citecolor=red,%
            filecolor=black,%
            linkcolor=blue,%
            urlcolor=blue}

\newtheorem{lemma}{Lemma}
\newtheorem{proposition}{Proposition}
\newtheorem*{proposition*}{Proposition}

\theoremstyle{definition}

\newtheorem{remark}{Remark}
\newtheorem*{remark*}{Remark}
\newtheorem{assumption}{Assumption}

\usepackage{pgffor}
\usepackage[dvipsnames]{xcolor}

\graphicspath{{img/}}

\usepackage{tikz}
\usetikzlibrary{shapes,positioning}

\tikzset{data/.style={draw,rectangle,rounded corners = 3pt,thick,fill=black!25}}
\tikzset{hp/.style={draw,rectangle}} 
\tikzset{mlink/.style={very thick,->,>=stealth}}
\tikzset{hlink/.style={<->,>=stealth,thick,dashed}}
\tikzstyle{grisEncadre}=[thick, dashed, fill=gray!20] 
\tikzstyle{grisFonce}=[fill=gray!100] 

\tikzstyle{worker}=[draw,rectangle,thick,rounded corners=3pt]
\tikzstyle{master}=[draw,ellipse,very thick,fill=black!25]
\tikzstyle{local_variable}=[draw,rectangle,thick,rounded corners=3pt,fill=black!10]

\def\firstToUp#1{\expandafter\firstToUpA#1!!} 
\def\firstToUpA#1#2!!{\MakeUppercase{#1}#2}

\usetikzlibrary{positioning,spy,calc}
\usepackage{pgfplots}
\pgfplotsset{compat=1.14}

\makeatletter
\let\MYcaption\@makecaption
\makeatother

\usepackage[font=footnotesize]{subcaption}

\makeatletter
\let\@makecaption\MYcaption
\makeatother

\newcommand\nbpix{N}
\newcommand\nendm{R}
\newcommand\nband{L}
\newcommand\ntime{\Omega}

\newcommand\nworker{\Omega}
\newcommand\iworker{\omega}
\newcommand\nmworker{K}

\def\Y{\mathbf{Y}}

\def\y{\mathbf{y}}

\def\M{\mathbf{M}}

\def\m{\mathbf{m}}

\def\A{\mathbf{A}}
\def\Atk{\A_{\iworker^k}}

\def\a{\mathbf{a}}

\def\X{\mathbf{X}}

\def\x{\mathbf{A}}

\def\B{\mathbf{B}}

\def\v{\mathbf{v}}
\def\w{\mathbf{w}}
\def\z{\mathbf{M}}

\def\niter{N_{\textrm{iter}}}

\def\dtk{d_{\iworker^k}^k}
\def\xtk{\x_{\iworker^k}}
\def\xhtk{\hat{\x}_{\iworker^k}}

\def\zh{\hat{\z}}
\def\ftk{f_{\iworker^k}}
\def\gtk{g_{\iworker^k}}

\def\Ldxz{L_{\x,\z}}

\def\Ldxtz{L_{\x_\iworker,\z}}

\def\Ldxtkz{L_{\x_{\iworker_k},\z}}

\def\Lx{L_{\x}}
\def\Lxtk{L_{\xtk}}

\def\Lz{L_{\z}}

\def\1#1{\mathbf{1}_{#1}}

\def\Fnorm#1{\norm{#1}_{\text{F}}^2}%
\def\enorm#1{\norm{#1}_{2}}%
\def\argmin#1{\underset{#1}{\arg \min \,}}

\def\t#1{#1^{\text{T}}}

\DeclareMathOperator{\aSAM}{aSAM}
\DeclareMathOperator{\GMSE}{GMSE}
\DeclareMathOperator{\RE}{RE}

\DeclareMathOperator{\prox}{prox}

\DeclarePairedDelimiter{\norm}{\bigl\lVert}{\bigr\rVert}
\DeclarePairedDelimiter{\pscalar}{\bigl\langle}{\bigr\rangle}

\title{Partially Asynchronous Distributed Unmixing of Hyperspectral Images}

\author{Pierre-Antoine Thouvenin,~\IEEEmembership{Member,~IEEE}, Nicolas Dobigeon,~\IEEEmembership{Senior Member,~IEEE} and Jean-Yves~Tourneret,~\IEEEmembership{Senior Member,~IEEE}

\thanks{This work was supported in part by the Hypanema ANR Project no. ANR-12-BS03-003, by the MapInvPlnt ERA-NET MED Project no. ANR-15-NMED-0002-02, by the Thematic Trimester on Image Processing of the CIMI Labex under Grant ANR-11-LABX-0040-CIMI within the Program ANR-11-IDEX-0002-02 and by the Direction G{\'e}n{\'e}rale de l'Armement, French Ministry of Defence.}
\thanks{This work has been conducted while P.-A. T. was working with the University of Toulouse. N. D. and J.-Y. T. are with the University of Toulouse, IRIT/INP-ENSEEIHT, 31071 Toulouse, France. (e-mail: pierreantoine.thouvenin@gmail.com, \{Nicolas.Dobigeon, Jean-Yves.Tourneret\}@enseeiht.fr}}

\begin{document}
\setlength{\textfloatsep}{4pt}
\setlength{\intextsep}{4pt}
\setlength{\abovedisplayskip}{5pt}
\setlength{\belowdisplayskip}{5pt}
\maketitle
\begin{abstract} 
So far, the problem of unmixing large or multitemporal hyperspectral datasets has been specifically addressed in the remote sensing literature only by a few dedicated strategies. Among them, some attempts have been made within a distributed estimation framework, in particular relying on the alternating direction method of multipliers (ADMM). In this paper, we propose to study the interest of a partially asynchronous distributed unmixing procedure based on a recently proposed asynchronous algorithm. Under standard assumptions, the proposed algorithm inherits its convergence properties from recent contributions in non-convex optimization, while allowing the problem of interest to be efficiently addressed. Comparisons with a distributed synchronous counterpart of the proposed unmixing procedure allow its interest to be assessed on synthetic and real data. Besides, thanks to its genericity and flexibility, the procedure investigated in this work can be implemented to address various matrix factorization problems.
\end{abstract}
\begin{IEEEkeywords}
Partially asynchronous distributed estimation, hyperspectral unmixing, non-convex optimization.
\end{IEEEkeywords}

\section{Introduction} \label{sec:intro}

\IEEEPARstart{A}{cquired} in hundreds of contiguous spectral bands, hyperspectral (HS) images present a high spectral resolution, which is mitigated by a lower spatial resolution in specific applications such as airborne remote sensing. The observed spectra are thus represented as mixtures of signatures corresponding to distinct materials. Spectral unmixing then consists in estimating the reference signatures associated with each material, referred to as endmembers, and their relative fractions in each pixel of the image, referred to as abundances, according to a predefined mixture model. In practice, a linear mixing model (LMM) is traditionally adopted when the declivity of the scene and microscopic interactions between the observed materials are negligible~\cite{Bioucas2012jstars}. Per se, HS unmixing can be cast as a blind source separation problem and, under the above assumptions, can be formulated as a particular instance of matrix factorization.

For this particular application, using distributed procedures can be particularly appealing to estimate the abundances since the number of pixels composing the HS images can be orders of magnitude larger than the number of spectral bands in which the images are acquired. In this context, distributed unmixing methods previously proposed in the remote sensing literature essentially rely on synchronous algorithms~\cite{Robila2013,Sigurdsson2016,Sigurdsson2017,Tsinos2017} with limited convergence guarantees. A different approach consists in resorting to a proximal alternating linearized minimization (PALM) \cite{Bolte2013,Chouzenoux2016} to estimate the mixture parameters (see, e.g., \cite{Repetti2014,Li2016,Thouvenin2015gretsi} in this context), which leads to an easily distributable optimization problem when considering the update of the abundances, and benefits from well established convergence results.

While a synchronous distributed variant of the PALM algorithm is particularly appealing to address HS unmixing, this algorithm does not fully exploit the difference in the computing performance of the involved computing units, which is precisely the objective pursued by the numerous asynchronous optimization techniques proposed in the optimization literature (e.g., \cite{Bianchi2013,Liang2014,Chen2015,Lorenzo2015,Facchinei2015,Scutari2017,Yang2016,Pesquet2014,Combettes2016}). For distributed synchronous algorithms, a master node waits for the information brought by all the available computation nodes (referred to as \emph{workers}) before proceeding to the next iteration (e.g., updating a variable shared between the different nodes, see Fig.~\ref{fig:sync}). On the contrary, asynchronous algorithms offer more flexibility in the sense that they allow more frequent updates to be performed by the computational nodes, thus reducing their idleness time. In particular, asynchronous algorithms can lead to a significant speed up in the algorithm computation time by allowing the available computational units (i.e., cores and machines) to work in parallel, with as few synchronizations (i.e., memory locks) as possible~\cite{Peng2016,Davis2016,Cannelli2016}. For some practical problems, there is no master node, and the workers can become active at any time and independently from the other nodes~\cite{Bianchi2016,Davis2016,Cannelli2016}. For other applications, a master node first assigns different tasks to all the available workers, then aggregates information from a given node as soon as it receives its information, and launches a new task on this specific node (see Fig.~\ref{fig:async}). In this partially asynchronous setting, the workers may make use of out-of-date information to perform their local updates~\cite{Combettes2016}. Given the possible advantages brought by the asynchronicity, we propose an asynchronous unmixing procedure based on recent non-convex optimization algorithms. To this end, we consider a centralized architecture as in~\cite{Chang2016}, composed of a master node in charge of a variable shared between the different workers, and $\ntime$ workers which have access to a local variable (i.e., only accessible from a given worker) and a (possibly out-of-date) local copy of the shared variable.

Asynchronous methods adapted to the aforementioned context include many recent papers, e.g., \cite{Lian2015,Chang2016,Peng2016,Davis2016,Cannelli2016}. For HS image unmixing, Gauss-Seidel optimization schemes have proved convenient to decompose the original optimization task into simpler sub-problems, which can be solved or distributed efficiently~\cite{Wright2015}. We may mention the recently proposed partially asynchronous distributed alternating direction method of multipliers (ADMM) \cite{Chang2016}, used to solve a distributed optimization task reformulated as a consensus problem. However, HS unmixing does not allow traditional block coordinate descent (BCD) methods (such as the ADMM~\cite{Boyd2010,Wang2016}) to be efficiently applied due to the presence of sub-problems which require iterative solvers. In such cases, the PALM algorithm~\cite{Bolte2013} and its extensions~\cite{Frankel2015,Chouzenoux2016}, which are sequential algorithms, combine desirable convergence guarantees for non-convex problems with an easily distributable structure in a synchronous setting. Recently, PALM has been extended to accommodate asynchronous updates~\cite{Davis2016}, and analyzed in a stochastic and a deterministic framework. More specifically, the author in~\cite{Davis2016} considers the general case where all the variables to be estimated are shared by the different workers. However, the explicit presence of a maximum allowable delay in the update steps is problematic, since this parameter is not explicitly controlled by the algorithm. In addition, the residual terms resulting from the allowed asynchronicity have a significant impact on the step-size prescribed to ensure the convergence of the algorithm. In practice, the use of this step-size does not lead to a reduction of the computation time needed to reach convergence, as it will be illustrated in Section~\ref{sec:exp}. From this practical point of view, the algorithm proposed in~\cite{Chang2016}, where the maximum delay is explicitly controlled, appears to be more convenient. However, the use of this ADMM-based algorithm does not ensure that the constraints imposed on the shared variables are satisfied at each iteration, and the sub-problems derived in the context of HS unmixing require the use of iterative procedures.
Finally, the strategy developed in~\cite{Cannelli2016} allows more flexibility in the allowed asynchronicity, while requiring  slightly more stringent assumptions on the penalty functions when compared to~\cite{Davis2016}.

Consequently, this paper proposes to adapt the framework introduced in \cite{Cannelli2016}, which encompasses the system structure described in \cite{Chang2016}, to HS unmixing. Indeed, given the preceding remarks, the framework introduced in \cite{Cannelli2016} appears as one of the most flexible to address HS unmixing in practice. This choice is partly justified by the possible connections between the PALM algorithm and~\cite{Cannelli2016}. Indeed, the PALM algorithm enables a synchronous distributed algorithm to be easily derived for matrix factorization problems, which then offers an appropriate reference to precisely evaluate the relevance of the asynchronicity tolerated by the approach described in~\cite{Cannelli2016}. Another contribution of this paper consists in assessing the interest of asynchronicity for HS unmixing, in comparison with recently proposed synchronous distributed unmixing procedures.

The paper is organized as follows. The problem addressed in this paper is introduced in Section~\ref{sec:LMM}. The proposed unmixing procedure is detailed in Section~\ref{sec:algorithm}, along with the assumptions required from the problem structure to recover appropriate convergence guarantees. Simulation results illustrating the performance of the proposed approach on synthetic and real data are presented in Sections~\ref{sec:exp} and~\ref{sec:real_exp}. Finally, Section~\ref{sec:conclusion} concludes this work and outlines possible research perspectives.


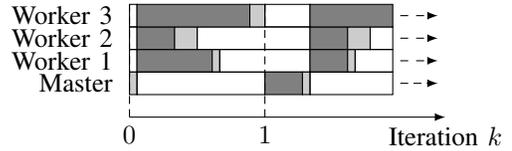
\begin{figure}[t]
\centering
\begin{tikzpicture}
\draw[black] (0,0.4) rectangle (0.1,0.7);
\draw[black,fill=gray!100](0.1,0.4) rectangle (1.1,0.7); 
\draw[black,fill=gray!40] (1.1,0.4) rectangle (1.2,0.7);
\draw[black] (1.2,0.4) rectangle (2.4,0.7);
\draw[black,fill=gray!100](2.4,0.4) rectangle (2.9,0.7);
\draw[black,fill=gray!40](2.9,0.4) rectangle (3,0.7);
\draw[black] (3,0.4) rectangle (3.5,0.7);
\draw[dashed,->,>=latex] (3.6,0.55) -- (4.1,0.55); 

\draw[black](0,0.7) rectangle (0.1,1);
\draw[black,fill=gray!100](0.1,0.7) rectangle (0.6,1); 
\draw[black,fill=gray!40] (0.6,0.7) rectangle (0.9,1);
\draw[black](0.9,0.7) rectangle (2.4,1);
\draw[black,fill=gray!100](2.4,0.7) rectangle (2.9,1);
\draw[black,fill=gray!40] (2.9,0.7) rectangle (3.2,1);
\draw[black](3.2,0.7) rectangle (3.5,1);
\draw[dashed,->,>=latex] (3.6,0.85) -- (4.1,0.85); 

\draw[black](0,1) rectangle (0.1,1.3);
\draw[black,fill=gray!100](0.1,1) rectangle (1.6,1.3); 
\draw[black,fill=gray!40] (1.6,1) rectangle (1.8,1.3);
\draw[black] (1.8,1) rectangle (2.4,1.3);
\draw[black,fill=gray!100](2.4,1) rectangle (3.5,1.3);
\draw[dashed,->,>=latex] (3.6,1.15) -- (4.1,1.15); 

\draw[black,fill=gray!40] (0,0.1) rectangle (0.1,0.4); 
\draw[black] (0.1,0.1) rectangle (1.8,0.4);
\draw[black,fill=gray!100] (1.8,0.1) rectangle (2.3,0.4);
\draw[black,fill=gray!40] (2.3,0.1) rectangle (2.4,0.4);
\draw[black](2.4,0.1) rectangle (3.5,0.4);
\draw[densely dashed,->,>=latex] (3.6,0.25) -- (4.1,0.25); 

\node[below] (k1)  at (1.8,-0.2) {$1$};
\draw[->,>=latex] (0,-0.2) -- (4.2,-0.2) node[below]{Iteration $k$} ;

\draw[densely dashed] (0,-0.2) -- (0,1.3);
\draw[densely dashed] (1.8,-0.2) -- (1.8,1.3);   

\node[below] (axis) at (0,-0.2) {$0$};
\node[left] (M)  at (-0.1,0.25) {Master};
\node[left] (W1) at (-0.1,0.55) {Worker 1};
\node[left] (W2) at (-0.1,0.85) {Worker 2};
\node[left] (W3) at (-0.1,1.15) {Worker 3};
\end{tikzpicture}
\caption{Illustration of a synchronous distributed mechanism (idle time in white, transmission delay in light gray, computation delay in gray). The master is triggered once it has received information from all the workers.}
\label{fig:sync}
\end{figure}

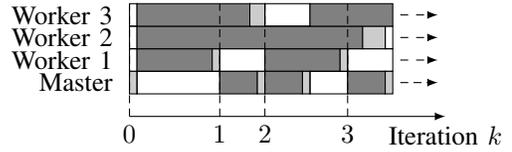
\begin{figure}[t]
\centering
\begin{tikzpicture}
\draw[black] (0,0.4) rectangle (0.1,0.7);
\draw[black,fill=gray!100](0.1,0.4) rectangle (1.1,0.7); 
\draw[black,fill=gray!40] (1.1,0.4) rectangle (1.2,0.7);
\draw[black] (1.2,0.4) rectangle (1.8,0.7);
\draw[black,fill=gray!100](1.8,0.4) rectangle (2.8,0.7);
\draw[black,fill=gray!40](2.8,0.4) rectangle (2.9,0.7);
\draw[black] (2.9,0.4) rectangle (3.5,0.7);
\draw[dashed,->,>=latex] (3.6,0.55) -- (4.1,0.55); 

\draw[black](0,0.7) rectangle (0.1,1);
\draw[black,fill=gray!100](0.1,0.7) rectangle (3.1,1); 
\draw[black,fill=gray!40] (3.1,0.7) rectangle (3.4,1);
\draw[black](3.4,0.7) rectangle (3.5,1);
\draw[dashed,->,>=latex] (3.6,0.85) -- (4.1,0.85); 

\draw[black](0,1) rectangle (0.1,1.3);
\draw[black,fill=gray!100](0.1,1) rectangle (1.6,1.3); 
\draw[black,fill=gray!40] (1.6,1) rectangle (1.8,1.3);
\draw[black] (1.8,1) rectangle (2.4,1.3);
\draw[black,fill=gray!100](2.4,1) rectangle (3.5,1.3);
\draw[dashed,->,>=latex] (3.6,1.15) -- (4.1,1.15); 

\draw[black,fill=gray!40] (0,0.1) rectangle (0.1,0.4); 
\draw[black] (0.1,0.1) rectangle (1.2,0.4);
\draw[black,fill=gray!100] (1.2,0.1) rectangle (1.7,0.4);
\draw[black,fill=gray!40] (1.7,0.1) rectangle (1.8,0.4);
\draw[black,fill=gray!100] (1.8,0.1) rectangle (2.3,0.4);
\draw[black,fill=gray!40] (2.3,0.1) rectangle (2.4,0.4);
\draw[black] (2.4,0.1) rectangle (2.9,0.4);
\draw[black,fill=gray!100] (2.9,0.1) rectangle (3.4,0.4);
\draw[black,fill=gray!40] (3.4,0.1) rectangle (3.5,0.4);
\draw[densely dashed,->,>=latex] (3.6,0.25) -- (4.1,0.25); 

\node[below] (k1)  at (1.2,-0.2) {$1$};
\node[below] (k2) at (1.8,-0.2) {$2$};
\node[below] (k3) at (2.9,-0.2) {$3$};
\draw[->,>=latex] (0,-0.2) -- (4.2,-0.2) node[below]{Iteration $k$} ;

\draw[densely dashed] (0,-0.2) -- (0,1.3);
\draw[densely dashed] (1.2,-0.2) -- (1.2,1.3);  
\draw[densely dashed] (1.8,-0.2) -- (1.8,1.3);  
\draw[densely dashed] (2.9,-0.2) -- (2.9,1.3);    

\node[below] (axis) at (0,-0.2) {$0$};
\node[left] (M)  at (-0.1,0.25) {Master};
\node[left] (W1) at (-0.1,0.55) {Worker 1};
\node[left] (W2) at (-0.1,0.85) {Worker 2};
\node[left] (W3) at (-0.1,1.15) {Worker 3};

%
\end{tikzpicture}
\caption{Illustration of an asynchronous distributed mechanism (idle time in white, transmission delay in light gray, computation delay in gray). The master node is triggered whenever it has received information from $\nmworker$ workers ($\nmworker~=~1$ in the illustration).}
\label{fig:async}
\end{figure}


\section{Problem formulation} \label{sec:LMM}

The LMM consists in representing each acquisition by a linear combination of the endmembers $\m_r$, which are present in unknown proportions. Assuming the data are composed of $\nendm$ endmembers, where $\nendm$ is \emph{a priori} known, and considering that the image is divided into $\ntime$ subsets of pixels (see Remark~\ref{remark1} for details) to distribute the data between several workers, the LMM can be defined as
\begin{equation} \label{eq:model}
    \Y_\iworker  = \M \A_\iworker + \mathbf{B}_\iworker, \; \iworker \in \{1, \dotsc, \ntime \}
\end{equation}
where $\Y_\iworker = \left[ \mathbf{y}_{1,\iworker},\dotsc,\mathbf{y}_{\nbpix,\iworker} \right]$ is an $\nband \times \nbpix$ matrix whose columns are the spectral signatures acquired for each pixel of the $\iworker$th pixel subset. Note that each group can be assigned a different number of pixels if needed. The columns $\m_r$ of the matrix $\M \in \mathbb{R}^{\nband \times \nendm}$ are the different endmembers, and the columns $\a_{n,\iworker}$ of the abundance matrix $\A_\iworker \in \mathbb{R}^{\nendm \times \nbpix}$ gather the proportion of the endmembers within $\y_{n,\iworker}$. Finally, the matrix $\B_\iworker \in \mathbb{R}^{\nband \times \nbpix}$ represents an additive noise resulting from the data acquisition and the modeling errors. The following constraints, aimed at ensuring a physical interpretability of the results, are usually considered
\begin{equation}
    \label{eq:constraints}
    \A_\iworker \succeq \mathbf{0}_{\nendm ,\nbpix}, \quad  \t{\A_\iworker} \mathbf{1}_\nendm  = \mathbf{1}_\nbpix, \quad \M \succeq \mathbf{0}_{\nband ,\nendm}
\end{equation}
where $\succeq$ denotes a term-wise inequality. Assuming the data are corrupted by a white Gaussian noise leads to the following data fitting term
\begin{equation}
    \label{eq:ft}
    f_\iworker(\A_\iworker,\M) = \frac{1}{2} \Fnorm{\Y_\iworker - \M \A_\iworker}.
\end{equation}
In addition, the constraints summarized in~\eqref{eq:constraints} are taken into account by defining
\begin{equation}
g_\iworker(\A_\iworker) = \iota_{\mathcal{A}_\nbpix} (\A_\iworker)
\end{equation}
\begin{align}
    & \mathcal{A}_\nbpix = \Bigl\{ \X \in \mathbb{R}^{\nendm \times \nbpix} \mid {\X}^T\mathbf{1}_\nendm = \mathbf{1}_\nbpix, \X \succeq \mathbf{0}_{\nendm, \nbpix}  \Bigr\} \\
    & r(\M) = \iota_{ \{\cdot \succeq \mathbf{0}\}} (\M)
\end{align}
where $\iota_\mathcal{S}$ denotes the indicator function of a set $\mathcal{S}$ ($\iota_\mathcal{S} (\mathbf{x}) = \mathbf{0}$ if $\mathbf{x} \in \mathcal{S}$, $+\infty$ otherwise). This leads to the following optimization problem
\begin{equation} \label{eq:problem}
    (\A^*, \M^*) \in \argmin{\A,\M} \Psi(\A, \M)
\end{equation}
with
\begin{align}
    & \Psi(\A,\M) = F(\A,\M) + G(\A) + r(\M) \\
    & F(\A,\M) = \sum_{\iworker = 1}^\nworker f_\iworker(\A_\iworker, \M) , \quad G(\A) = \sum_{\iworker = 1}^\nworker g_\iworker (\A_\iworker).
\end{align}

With these notations, $\A_\iworker$ denotes a  \emph{local} variable (i.e., which will be accessed by a single worker), and $\M$ is a global variable (i.e., shared between the different workers, see Fig.~\ref{fig:architecture}). More generally, $f_\iworker$ plays the role of a data fitting term, whereas $g_{\iworker}$ and $r$ can be regarded as regularizers or constraints. The structure of the proposed unmixing algorithm, inspired by~\cite{Cannelli2016}, is detailed in the following section.

\begin{remark} \label{remark1}
    In the initial formulation of the mixing model~\eqref{eq:model}, the indexes $\iworker$ and $\nworker$ refer to subsets of pixels. A direct interpretation of this statement can be obtained by dividing a unique (and possibly large) hyperspectral image into $\nworker$ non-overlapping tiles of smaller (and possibly different) sizes. In this case, each tile is individually unmixed by a given worker. Another available interpretation allows multitemporal analysis to be conducted. Indeed, in practice, distributed unmixing procedures are of particular interest when considering the unmixing of a sequence of several HS images, acquired by possibly different sensors at different dates, but sharing the same materials~\cite{Henrot2016,Thouvenin2015b,Yokoya2017}. In this case, $\iworker$ and $\nworker$ could refer to time instants. Each worker $\iworker$ is then dedicated to the unmixing of a unique HS image acquired at a given time instant. The particular applicative challenge of distributed unmixing of multitemporal HS images partly motivates the numerical experiments on synthetic (yet realistic) and real data presented hereafter.
\end{remark}

\begin{remark}
Even if the work reported in this work has been partly motivated by the particular application of HS unmixing, the problem formulated in this section is sufficiently generic to encompass a wider class of matrix factorization tasks, as those encountered in audio processing \cite{Fevotte2007}, machine learning \cite{Tan2013,Gao2014}.
\end{remark}


\section{A partially asynchronous unmixing algorithm} \label{sec:algorithm}

    \subsection{Algorithm description}
    
    \begin{figure}
\centering
    \begin{tikzpicture}
        \node[master] (M)at(0,0){Master};
        \node[worker,below left = 1 and 2 of M.center] (W1) {Worker 1};
        \node[worker,below = 1 of M.center] (W2) {Worker 2};
        \node[worker,below right = 1 and 2 of M.center] (W3) {Worker 3};
        \node[local_variable,below = 0.5 of W1] (V1) {$f_1, g_1, \A_1$}; 
        \node[local_variable,below = 0.5 of W2] (V2) {$f_2, g_2, \A_2$};
        \node[local_variable,below = 0.5 of W3] (V3) {$f_3, g_3, \A_3$};
        \node[local_variable,above = 0.5 of M] (Z) {$F, G, r, \M$}; 
        %
        \draw[<->,>=latex,dashed] (M.west) to[in=90,out=-180] (W1.north);
        \draw[<->,>=latex,dashed] (M.south) -- (W2.north);
        \draw[<->,>=latex,dashed] (M.east) to[in=90,out=0] (W3.north);
        %
        \draw[very thick] (W1.south) --  (V1.north);
        \draw[very thick] (W2.south) --  (V2.north);
        \draw[very thick] (W3.south) --  (V3.north);
        \draw[very thick] (M.north) --  (Z.south);
    \end{tikzpicture}
    \caption{Illustration of the master-slave architecture considered for the unmixing problem \eqref{eq:problem} with $\Omega = 3$ workers (the function and variables available at each node are given in light gray rectangles).}
    \label{fig:architecture}
\end{figure}
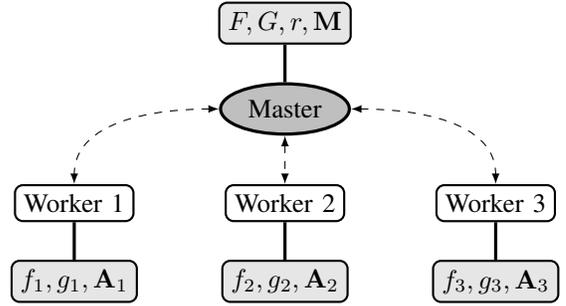

Reminiscent of~\cite{Chang2016}, the proposed algorithm relies on a star topology configuration in which a master node supervises an optimization task distributed between several workers. The master node also updates and transmits the endmember matrix $\M$ shared by the different workers. In fact, the computation time of synchronous algorithms is essentially conditioned by the speed of the slowest worker (see Figs.~\ref{fig:sync} and~\ref{fig:async}). Consequently, relaxing the synchronization requirements (by allowing bounded delays between the information brought by each worker) allows a significant decrease in the computation time to reach convergence, which can scale almost linearly with the number of workers~\cite{Chang2016,Davis2016}. Note that, even though asynchronous optimization schemes may require more iterations than their synchronous counterparts to reach a given precision, allowing more frequent updates generally compensates this drawback in terms of computation time~\cite{Chang2016}.

In the partially asynchronous setting considered, the master node updates the variable shared by the workers once it has received information from at least $\nmworker \ll \ntime$ workers. The new state of the shared variable $\M$ is then transmitted to the $\nmworker$ available workers, which can individually proceed to the next step. As in~\cite{Cannelli2016}, a relaxation step with decreasing stepsizes ensures the convergence of the algorithm (see Algo.~\ref{alg:master}). In order to clarify to which extent the convergence analysis introduced in~\cite{Cannelli2016} is applicable to the present setting, we consider $\nmworker = 1$ in the rest of this paper. However, other values of $K$ could be considered without loss of generality. Details on the operations performed by the master node and each worker are detailed in Algos.~\ref{alg:master} and~\ref{alg:worker} respectively.

\begin{remark}
\begin{enumerate}[label=(\alph*)]
\item The parameter $\gamma_k$ is essentially instrumental to ensure the global convergence of the partially asynchronous unmixing algorithm described in this work, following the general framework introduced in \cite{Cannelli2016}. For simplicity, we have directly adopted the expression proposed in \cite{Scutari2017} \cite[Assumption D., p. 18]{Cannelli2016} which has been reported to yield satisfactory results in practice \cite{Scutari2017}. Evaluating the practical interest of different expressions for the relaxation parameters in terms of the convergence speed of the algorithm is an interesting prospect, which is however beyond the scope of this paper.
\item Note that a synchronous distributed counterpart of Algo.~\ref{alg:master} can be easily derived for Problem~\eqref{eq:problem}, which partly justifies the form chosen for Algo.~\ref{alg:master}. This version consists in setting $\gamma_k = 1$, and waiting for the updates performed by all the workers (i.e., $\nmworker = \nworker$, see~\ref{alg:wait_step} of Algo.~\ref{alg:master}) before updating the shared variable $\M$. This implementation will be taken as a reference to evaluate the computational efficiency of the proposed algorithm in Sections~\ref{sec:exp} and~\ref{sec:real_exp}. 
\end{enumerate}
\end{remark}

\begin{algorithm}[!t]
\footnotesize{
 \KwData{$\x^{(0)}$, $\z^{(0)}$, $\gamma_0 \in (0,1]$, $\mu \in (0,1)$, $\niter$, $K$.}
Broadcast $\z^{(0)}$ to the $\ntime$ workers \;
$k \leftarrow 0$ \;
$\mathcal{T}_k \leftarrow \emptyset$ \;
\While{$k < \niter$}{
	\nlset{Step 1}\label{alg:wait_step}Wait for $\xhtk^k$ from any worker \;
	$\mathcal{T}_k = \mathcal{T}_k \cup \{ \iworker^k \}$ \;
   	$d_\iworker^{k+1} = \left\{
   	\begin{array}{l}
		0 \text{ if } \iworker \in \mathcal{T}_k \\
		d_\iworker^k + 1 \text{ otherwise}
   	\end{array} \right.$\medskip\;
   	$\x_\iworker^{k+1} = \left\{ \begin{array}{l}
   	\x_\iworker^k + \gamma_k (\hat{\x}_\iworker - \x_\iworker^k) \text{ if } \iworker \in \mathcal{T}_k \\
   	\x_\iworker^k \text{ otherwise}
   	\end{array} \right.$\medskip\;
   	\If{$(\sharp \mathcal{T}_k < K)$}{
   		Go to step \ref{alg:wait_step} \tcp*{wait until $\sharp \mathcal{T}_k \geq K$}
   	}
   	\Else{
   		$ \displaystyle \hat{\z}^k \in \prox_{c_{\z}^k r} \Bigl( \z^k + \frac{1}{c_{\z}^k} \nabla_{\z} F(\x^{k+1},\z^k) \Bigr)$\;
   		$\z^{k+1} = \hat{\z}^k + \gamma_k (\hat{\z}^k - \z^k)$\;
   		$\gamma_{k+1} = \gamma_k (1 - \mu\gamma_k)$\;
   		$\mathcal{T}_{k+1} \leftarrow \emptyset$ \;
   		$k \leftarrow k+1$\;
   	}
}
\KwResult{$\x^{\niter}$, $\z^{\niter}$.}
}
\caption{Master node update.} 
\label{alg:master}
\end{algorithm}

    \subsection{Parameter estimation} \label{sec:parameter_estimation}

A direct application of the algorithm described in Algo.~\ref{alg:worker} under the constraints~\eqref{eq:constraints} leads to the following update rule for the abundance matrix $\A_{\iworker^k}$
\begin{align}
    \label{eq:update_A}
    \widehat{\A}_{\iworker^k}^k &= \prox_{\iota_{\mathcal{A}_\nbpix}} \biggl( \A_\iworker^{k} - \frac{1}{c_{\Atk}^k} \nabla_{\A_\iworker} f_\iworker \bigl(\A_{\iworker^k}^k,\M^{k - \dtk} \bigr)\biggr)
\end{align}
where $\prox_{\iota_{\mathcal{A}_\nbpix}}$ denotes the proximal operator of the indicator function $\iota_{\mathcal{A}_\nbpix}$ (see, e.g., \cite{Combettes2011}), and
\begin{equation}
    \nabla_{\A_\iworker} f_\iworker(\A_\iworker,\M) = \t{\M} \bigl(\M\A_\iworker - \Y_\iworker \bigr).
\end{equation}
The step-size $c_{\Atk}^k$  is chosen as in the standard PALM algorithm, i.e.,
\begin{equation}
    c_{\Atk}^k = L_{\Atk}^k = \enorm{\t{(\M^{k - \dtk})} \M^{k - \dtk}}
\end{equation}
where $L_{\Atk}^k$ denotes the Lipschitz constant of $\nabla_{\A_\iworker} f_\iworker(\cdot,\M^{k-\dtk})$ (see \cite[Remark 4 (iv)]{Bolte2013}).
Note that the projection $\prox_{\iota_{\mathcal{A}_\nbpix}}(\cdot)$ can be exactly computed (see~\cite{Duchi2008,Condat2015} for instance). Similarly, the update rule for the endmember matrix $\M$ is
\begin{equation}
    \label{eq:update_M}
    \widehat{\M}^k = \prox_{\iota_{\{ \cdot \succeq \mathbf{0} \}}} \biggl(\M^k - \frac{1}{c_\M^k} \nabla_{\M} F\bigl( \A^{k+1},\M^k \bigr) \biggr)
\end{equation}
with
\begin{equation}
    \nabla_{\M} F \bigl(\A,\M \bigr) = \sum_\iworker (\M \A_\iworker - \Y_\iworker) \t{\A_\iworker}
\end{equation}
\begin{equation}
    c_\M^k = L_\M^k = \enorm{\sum_\iworker \A_\iworker^{k+1} \t{(\A_\iworker^{k+1})}}.
\end{equation}
and $L_\M^k$ is the Lipschitz constant of $\nabla_{\M} F \bigl(\A^k,\cdot \bigr)$.

\begin{algorithm}[!t]
\footnotesize{
\KwData{$\tilde{\z}$, $\tilde{\x}_\iworker$.}
\Begin{
	 Wait for $(\tilde{\z},\tilde{\x}_\iworker)$ from the master node\;  	
	     $\displaystyle \hat{\x}_\iworker \in \prox_{c_{\x_\iworker} g_\iworker} \Bigl( \tilde{\x}_\iworker - \frac{1}{c_{\x_\iworker}} \nabla_{\x_\iworker} f_\iworker \bigl(\tilde{\x}_\iworker,\tilde{\z} \bigr) \Bigr)$\;
	Transmit $\hat{\x}_\iworker$ to the master node\;		  		
}
\KwResult{$\hat{\x}_\iworker$.}
}
\caption{$\iworker$th worker update (since the shared variable $\z$ may have been updated by the master node in the meantime, $\tilde{\z}$ corresponds to a possibly delayed version of the current $\z^k$). From the master point of view, $\tilde{\z} = \z^{k-d_\iworker^k}$.}
\label{alg:worker}
\end{algorithm} 

    \subsection{Convergence guarantees} \label{subsec:convergence}

In general, the proposed algorithm requires the following assumptions, based on the convergence results given in \cite[Theorem 1]{Bolte2013} and \cite[Theorem 1]{Cannelli2016}.

\begin{assumption}[Algorithmic assumption] \label{alg_assumption}
    Let $(\iworker_k,\dtk) \in \{1,\dotsc,\ntime \}\times \{1,\dotsc,\tau \}$ denote the couple composed of the index of the worker transmitting information to the master at iteration $k$, and the delay between the (local) copy $\tilde{\M}^k$ of the endmember matrix $\M$ and the current state $\M^k$ (i.e., $\tilde{\M}^k \triangleq \M^{k - \dtk}$). The allowable delays $\dtk$ are assumed to be bounded by a constant $\tau \in \mathbb{N}^*$. In addition, each couple $(\iworker_k,\dtk)$ represents a realization of a random vector within the probabilistic model introduced in \cite[Assumption C]{Cannelli2016}.
\end{assumption}

\begin{assumption}[Inherited from PALM~\cite{Bolte2013}] \label{assumption}~
    \begin{enumerate}[label=(\roman*)]
        \item For any $\iworker \in \{1, \dotsc, \ntime\}$, $g_\iworker : \mathbb{R}^{\nendm \times \nbpix} \rightarrow(-\infty, +\infty]$ and $r : \mathbb{R}^{\nband \times \nendm} \rightarrow(-\infty, +\infty]$ are proper, convex lower semi-continuous (l.s.c.) functions; \label{assumption_convexity}
        \item For $\iworker \in \{1, \dotsc, \ntime\}$, $f_\iworker : \mathbb{R}^{\nendm \times \nbpix} \times \mathbb{R}^{\nband \times \nendm} \rightarrow \mathbb{R}$ is a $\mathcal{C}^1$ function, and is convex with respect to each of its variables when the other is fixed;
        \item $\Psi$, $f_\iworker$, $g_\iworker$, and $r$ are lower bounded, i.e., $\inf_{\mathbb{R}^{\nendm \times \nbpix} \times \mathbb{R}^{\nband \times \nendm}} \Psi > -\infty$, $\inf_{\mathbb{R}^{\nendm \times \nbpix} \times \mathbb{R}^{\nband \times \nendm}} f_\iworker > -\infty$, $\inf_{\mathbb{R}^{\nendm \times \nbpix}} g_\iworker > -\infty$, and $\inf_{\mathbb{R}^{\nband \times \nendm}} r > -\infty$;
        \item $\Psi$ is a coercive semi-algebraic function (see~\cite{Bolte2013}); \label{assumption_coercivity}
        \item For all $\iworker \in \{1, \dotsc, \ntime\}$, $\z \in \mathbb{R}^{\nband \times \nendm}$, $\x_\iworker~\mapsto~f_\iworker(\x_\iworker,\z)$ is a $\mathcal{C}^1$ function, and the partial gradient $\nabla_{\x_\iworker} f_\iworker(\cdot,\z)$ is Lipschitz continuous with Lipschitz constant $L_{\x_\iworker}(\z)$. Similarly, $\z \mapsto f_\iworker(\x_\iworker,\z)$ is a $\mathcal{C}^1$ function, and the partial gradient $\nabla_{\z} f_\iworker(\x_\iworker,\cdot)$ is Lipschitz continuous, with constant $L_{\z,\iworker}(\x_\iworker)$; \label{assumption:partial_grad}
        \item the Lipschitz constants used in the algorithm, i.e., $L_{\xtk^k}(\tilde{\z}^k)$ and $L_{\z,\iworker_k}(\xhtk^k)$ (denoted by $L_{\xtk^k}^k$ and $L_{\z,\iworker_k}^k$ in the following) are bounded, i.e. there exists appropriate constants such that for all iteration index $k$ \label{assumption_lip}
        \begin{equation*}
            0 < L_{\x}^- \leq L_{\xtk}^k \leq L_{\x}^+, \quad 0 < L_{\z}^- \leq L_{\z,\iworker^k}^k \leq L_{\z}^+.
        \end{equation*}
        \item $\nabla F$ is Lipschitz continuous on bounded subsets. \label{assumption_lip_bounded}
    \end{enumerate}
\end{assumption}

\begin{assumption}[Additional assumptions] \label{assumption2}
    \begin{enumerate}[label=(\roman*)]
        \item For all $\iworker \in \{1, \dotsc, \ntime\}$, $\x_\iworker \in \mathbb{R}^{\nendm \times \nbpix}$,  $\nabla_{\x_\iworker} f_\iworker(\x_\iworker,\cdot)$ is Lipschitz continuous with Lipschitz constant $\Ldxtz(\x_\iworker)$;
        \item The Lipschitz constants $\Ldxtkz(\xhtk^k)$ (denoted by $\Ldxtkz^k$ in the following) is bounded, i.e. there exists appropriate positive constants such that for all $k \in \mathbb{N}$:\label{assumption2_lip}
            \begin{equation*}
            0 < \Ldxz^- \leq \Ldxtkz^k \leq \Ldxz^+.
            \end{equation*}
    \end{enumerate}
\end{assumption}

Assumption~\ref{alg_assumption} summarizes standard algorithmic assumptions to ensure the convergence of Algo.~\ref{alg:master}. Besides, Assumption~\ref{assumption} gathers requirements of the traditional PALM algorithm~\cite{Bolte2013}, under which the distributed synchronous version of the proposed algorithm can be ensured to converge.

Note that the non-convex problem~\eqref{eq:problem} obviously satisfies Assumptions~\ref{assumption} to~\ref{assumption2} for the functions defined in Section~\ref{sec:LMM} (see~\cite{Bolte2013} for examples of semi-algebraic functions). In particular, the bounds on the Lipschitz constants involved in Assumptions~\ref{assumption}\ref{assumption_lip_bounded} and~\ref{assumption2}\ref{assumption2_lip} are satisfied in practice, considering the fact that hyperspectral unmixing is generally conducted on reflectance data (implying $\Y_\omega \in [0, 1]^{\nband \times \nbpix}$), and given the constraints imposed on $\A_\omega$ and $\M$ respectively.


Under Assumptions~\ref{alg_assumption} to~\ref{assumption2}, the analysis led in~\cite{Cannelli2016} allows the following convergence result to be satisfied.

\begin{proposition} \label{prop_stochastic}
    Suppose that Problem~\eqref{eq:problem} satisfies the requirements specified in Assumptions~\ref{alg_assumption} to~\ref{assumption2}. Define the sequence $\{ \v^k \}_{k \in \mathbb{N}}$ of the iterates generated by Algos.~\ref{alg:master} and~\ref{alg:worker}, with $\v^k \triangleq (\A^k,\M^k)$ and the parameters in Algo.~\ref{alg:worker} chosen as
    \begin{equation*}
        c_{\xtk}^k = \Lxtk^k, \quad c_\z^k = \Lz^k.
    \end{equation*}
    Then, the following convergence results are obtained:
    \begin{enumerate}[label=(\roman*)]
        \item the sequence $\{ \Psi(\v^k) \}_{k \in \mathbb{N}}$ converges almost surely;
        \item every limit point of the sequence $\{ \v^k \}_{k \in \mathbb{N}}$ is a critical point of $\Psi$ almost surely.
    \end{enumerate}
\end{proposition}

\begin{proof}
See sketch of proof in the Appendix.
\end{proof}

The convergence analysis is conducted using an auxiliary function (introduced in Lemma~\ref{lemma2} in Appendix) to handle asynchronicity~\cite{Davis2016}. The resulting convergence guarantees then allow convergence results associated with the original problem~\eqref{eq:problem} to be recovered.

Besides, the following result ensures a stronger convergence guarantee for the synchronous counterpart of Algo.~\ref{alg:master}.

\begin{proposition}[Finite length property, following from~\cite{Bolte2013}] \label{prop_deterministic}
    Suppose that Problem~\eqref{eq:problem} satisfies the requirements specified in Assumptions~\ref{assumption} to~\ref{assumption2}. Define the sequence $\{ \v^k \}_{k \in \mathbb{N}}$ of the iterates generated by the synchronous version of Algo.~\ref{alg:master}, with $\v^k \triangleq (\x^k,\z^k)$ and
    \begin{equation*}
        c_{\xtk}^k = \Lxtk^k, \quad c_\z^k = \Lz^k, \quad \gamma_k = 1, \quad \nmworker = \nworker.
    \end{equation*}
    Then, the following properties can be proved:
    \begin{enumerate}[label=(\roman*)]
        \item the sequence $\{ \v^k \}_{k \in \mathbb{N}}$ has finite length
        \begin{equation*}
            \begin{split}
                &\sum_{k=1}^{+\infty} \norm{\v^{k+1} - \v^k} < +\infty
            \end{split}
        \end{equation*}
        where
        \begin{equation*}
            \norm{\v^{k+1} - \v^k} = \sqrt{\Fnorm{\A^{k+1} - \A^k} + \Fnorm{\M^{k+1} - \M^k}};
        \end{equation*}
        \item the sequence $\{ \v^k \}_{k \in \mathbb{N}}$ converges to a critical point of $\Psi$.
    \end{enumerate}
\end{proposition}

\begin{proof}
These statements result from a direct application of \cite[Theorem 1, Theorem 3]{Bolte2013} and \cite[Remark 4 (iv)]{Bolte2013}.
\end{proof}

Note that an additional volume regularization can be considered, as long as it satisfies the conditions given in~Assumption~\ref{assumption}, and more specifically the convexity Assumption~\ref{assumption}\ref{assumption_convexity}. For instance, the mutual distance between the endmembers introduced in \cite{Berman2004} can be easily accounted for.

\section{Experiments with synthetic data} \label{sec:exp}

To illustrate the interest of the allowed asynchronicity, we compare the estimation performance of Algo.~\ref{alg:master} to the performance of its synchronous counterpart (described in Section~\ref{sec:algorithm}), and evaluate the resulting unmixing performance in comparison with three unmixing methods proposed in the literature. We propose to consider the context of multitemporal HS unmixing, which is of particular interest for recent remote sensing applications~\cite{Henrot2016,Thouvenin2015b,Yokoya2017}. For this application, a natural way of distributing the data consists in assigning a single HS image to each worker.
To this end, we generated synthetic data composed of $\ntime = 3$ HS images resulting from linear mixtures of $\nendm \in \left\{3, 6, 9\right\}$ endmembers acquired in $\nband = 413$ bands. The generated abundance maps vary smoothly over time (i.e., from one image to another) to reproduce a realistic evolution of the scene of interest. As in \cite[Section V]{Thouvenin2018}, the abundance maps were obtained by multiplying reference abundance coefficients with trigonometric functions to ensure a sufficiently smooth temporal evolution. For the dataset with $\nendm = 3$, the reference abundance map was obtained by unmixing the Moffett scene (same area as in \cite{Dobigeon2009}). For the datasets composed of $\nendm \in \{6, 9\}$ endmembers, we directly used the synthetic abundance maps introduced in \cite{Plaza2011} as a reference\footnote{Abundance maps available at \url{http://www.umbc.edu/rssipl/people/aplaza/fractals.zip}.}. Each image, composed of $10,000$ pixels, was then corrupted by an additive white Gaussian noise whose variance ensures a signal-to-noise ratio (SNR) of $30$ dB.

Note that the distributed methods were run on a single computer for illustration purposes using the built-in low level distributed computing instructions available in Julia \cite{Bezanson2017} (which provide an MPI-like interface). In this case, the workers are independent processes. 

As is common with many blind unmixing algorithms, the performance of the proposed approach is expected to be limited in cases where the initial endmember matrix does not properly represent the observed materials. This observation essentially results from the nonconvex nature of the problem presently addressed, and is not specific to the proposed approach. To the best of the authors' knowledge, no blind unmixing algorithm can systematically ensure the convergence of the generated iterates to a ``satisfactory'' critical point of the objective function in cases where the initialization is relatively poor.

\begin{figure*}[thp]
\centering
\foreach \i [count = \r] in {3,6,9} {
	\begin{subfigure}[t]{0.32\textwidth}
		\centering
		\includegraphics[keepaspectratio,width=0.99\textwidth]{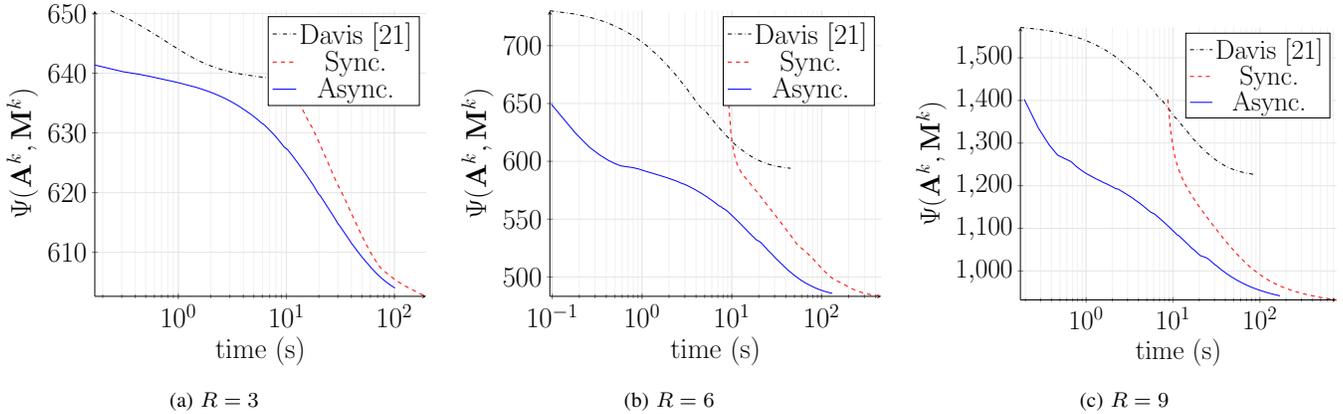}
		\caption{$\nendm = \i$}
		\label{fig:f_r\i}
	\end{subfigure}
	}
\caption{Evolution of the objective function for the synthetic datasets, obtained for Algo.~\ref{alg:master} and its synchronous version until convergence (model \eqref{eq:model}).}
\label{fig:objective}
\end{figure*}


    \subsection{Compared methods}

The estimation performance of the proposed algorithm has been compared to those of several unmixing methods from the literature. Note that only the computation times associated with Algo.~\ref{alg:master} and its synchronous version, implemented in Julia~\cite{Bezanson2017}, can lead to a consistent comparison in this experiment. Indeed, some of the other unmixing methods have been implemented in \textsc{Matlab} by their respective authors. In the following lines, implementation details specific to each of these methods are given.
\begin{enumerate}
    \item VCA/FCLS: the endmembers are first extracted on each image using the vertex component analysis (VCA) \cite{Nascimento2005}, which requires pure pixels to be present. The abundances are then estimated for each pixel by solving a fully constrained least squares problem (FCLS) using the ADMM algorithm described in~\cite{Bioucas2010};
    \item SISAL/FCLS: the endmembers are extracted on each image by the simplex identification via split augmented Lagrangian (SISAL) \cite{Bioucas2009}, and the abundances are estimated for each pixel by FCLS. The tolerance for the stopping rule is set to $10^{-4}$;
    \item Proposed method (referred to as ASYNC): the endmembers are initialized with the signatures obtained by VCA on the first image of the sequence, and the abundances are initialized by FCLS. The synchronous and asynchronous algorithms are stopped when the relative decrease of the objective function between two consecutive iterations is lower than $10^{-5}$, with a maximum of 100 and 500 iterations respectively. Its synchronous counterpart is referred to as SYNC. The relaxation parameter $\gamma_k$ $(k \in \mathbb{N}^*)$ is updated as in \cite{Cannelli2016} with $\gamma_0 = 1$ and $\mu = 10^{-6}$ (see Algo. ~\ref{alg:master}). In the absence of any temporal or spatial reglarization, the lexicographically ordered pixels composing the datasets are evenly distributed between $\nworker = 3$ workers;
    \item DAVIS~\cite{Davis2016}: this asynchronous algorithm only differs from the previous algorithm, in that no relaxation step is considered, and in the expression of the descent stepsize used to ensure the algorithm convergence. To ensure a fair comparison, it has been run in the same setting as the proposed asynchronous method;
    \item DSPLR~\cite{Tsinos2017}: the DSPLR algorithm is considered with the stopping criterion proposed in~\cite{Tsinos2017} (set to $\varepsilon = 10^{-5}$), with a maximum of 100 iterations. The same initialization as the two previous distributed algorithms is used.
\end{enumerate}

\setlength\columnsep{0.1pt}
\begin{table}[t!] 
\centering
\caption{Simulation results on synthetic data (GMSE($\mathbf{A}$)$\times 10^{-3}$, RE $\times 10^{-4}$).}
	\begin{center}
	\resizebox{0.48\textwidth}{!}{%
		\begin{tabular}{@{}lllccccc@{}} \toprule
&	Algorithm	&   	   & aSAM($\M$) (\textdegree) & GMSE($\A$) & RE & aSAM($\Y$) (\textdegree) & time (\si{\second}) \\ \cmidrule{1-8}
\multirow{6}{*}{\rotatebox{90}{$\nendm = 3$}}
&VCA/FCLS &\cite{Nascimento2005}      & 1.82 & 1.27 & 0.64 & 1.45 & \textbf{1} \\
&SISAL/FCLS	&\cite{Bioucas2009}   & 1.55 & 0.94 & 0.62 & 1.43 & 2   \\
&DSPLR &\cite{Tsinos2017} &	0.84 &	2.76 &	\textbf{0.59} &	\textbf{1.41} &	139 \\
&DAVIS &\cite{Davis2016} &	1.44 &	0.92 &	0.63 &	1.45 &	10 \\ 
&SYNC &	&\textbf{0.76} &	\textbf{0.33} &	0.60 &	\textbf{1.41} &	197 \\
&ASYNC &	&0.85 &	0.38 &	0.60 &	\textbf{1.41} &	101 \\
\cmidrule{1-8}
\multirow{6}{*}{\rotatebox{90}{$\nendm = 6$}}
&VCA/FCLS &\cite{Nascimento2005}      & 2.55 & 1.08 & 1.11 & 1.64 & \textbf{1} \\
&SISAL/FCLS &\cite{Bioucas2009}	   & 1.65 & 0.50 & \textbf{0.91} & 1.53 & 2.5 \\
&DSPLR &\cite{Tsinos2017} &	3.64 &	4.65 &	7.73 &	\textbf{1.45} &	116 \\
&DAVIS &\cite{Davis2016} &	1.87 &	1.22 &	0.96 &	1.58 &	45.3 \\
&SYNC &	&\textbf{0.63} &	\textbf{0.28} &	\textbf{0.78} &	\textbf{1.45} &	462 \\
&ASYNC &	&1.09 &	0.59 &	0.81 &	1.48 &	46 \\
\cmidrule{1-8}
\multirow{6}{*}{\rotatebox{90}{$\nendm = 9$}}
&VCA/FCLS  &\cite{Nascimento2005}    & 3.07 & 2.59 & 6.75 & 2.37 & \textbf{2} \\
&SISAL/FCLS	&\cite{Bioucas2009}   & 2.17 & 1.77 & 5.11 & 2.14 & 4   \\
&DSPLR &\cite{Tsinos2017} &	8.52 &	6.53 &	1.48 &	\textbf{1.56} &	153 \\
& DAVIS &\cite{Davis2016} &	1.57 &	1.27 &	1.98 &	1.69 &	84 \\ 
&SYNC &	&\textbf{0.87} &	\textbf{0.40} &	\textbf{1.50} &	1.57 &	762 \\
&ASYNC &	&0.88 &	0.54 &	1.52 &	1.58 &	170 \\
\bottomrule
		\end{tabular}
}
	\end{center}
\label{tab:results_synth}
\end{table}


The estimation performance reported in Table~\ref{tab:results_synth} are evaluated in terms of
\begin{enumerate}[label=(\roman*)] 
    \item endmember estimation and spectral reconstruction through the average spectral angle mapper (aSAM)
\end{enumerate}
\begin{equation}
    \aSAM(\M) = \frac{1}{\nendm } \sum_{r=1}^\nendm   \arccos \left( \frac{ \mathbf{m}_r^T \widehat{\mathbf{m}}_r }{ \lVert \mathbf{m}_r \rVert_2 \lVert \widehat{\mathbf{m}}_r \rVert_2 } \right)
\end{equation}
\begin{equation} \label{eq:aSAM_Y}
    \aSAM(\Y) = \frac{1}{\nbpix \ntime} \sum_{n,\iworker} \arccos \left( \frac{ \mathbf{y}_{n,\iworker}^T \bigl(\widehat{\M} \hat{\a}_{n,\iworker} \bigr) }{ \lVert \mathbf{y}_{n,\iworker} \rVert_2 \lVert \widehat{\M} \hat{\a}_{n,\iworker} \rVert_2 } \right);
\end{equation}
\begin{enumerate}[label=(\roman*),resume]
    \item abundance estimation through the global mean square error (GMSE) \vspace{-0.2cm}
    \begin{align}
        \GMSE(\A)  & = \frac{1}{\ntime \nendm \nbpix} \sum_{\iworker=1}^\ntime \lVert \A_\iworker - \widehat{\A}_\iworker \rVert_{\text{F}}^2;
    \end{align}
    \item quadratic reconstruction error (RE)
    \begin{align}
        \label{eq:RE}
        \RE &= \frac{1}{\ntime \nband \nbpix} \sum_{\iworker=1}^\ntime \lVert \Y_\iworker - \widehat{\M} \widehat{\A}_\iworker \rVert_{\text{F}}^2.
    \end{align}
\end{enumerate}

\def\names{{04/10/2014},{06/02/2014},{09/19/2014},{11/17/2014},{04/29/2015},{10/13/2015}}

\begin{figure*}
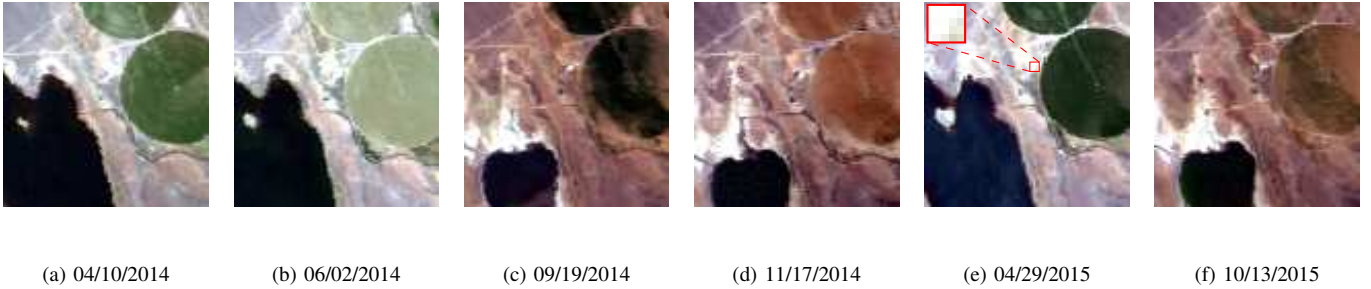

\centering
\foreach \name[count=\j] in \names {
	\def\idx{\the\numexpr\j+3}    
	\ifnum\j=5
		\begin{subfigure}[t]{0.15\textwidth}
			\begin{tikzpicture}
    \begin{scope}[
    node distance = 2.5mm,
        inner sep = 0pt,spy using outlines={rectangle, red, magnification=4, 
        }
                 ]
\node (img)  {\includegraphics[keepaspectratio,width=\textwidth]{\idx}};
\coordinate (F) at (0.1, 0.5);
\coordinate (G) at (-1.5, 1.5);
\spy [size=5mm] on (F) in node[below right=of G.north east];
\end{scope}
\draw[dashed,red] (tikzspyonnode.north east) -- (tikzspyinnode.north east);
\draw[dashed,red] (tikzspyonnode.south west) -- (tikzspyinnode.south west);
    \end{tikzpicture}
    		\caption{\name}
    		\label{fig:cube\j}
    	\end{subfigure}	
	\else
		\begin{subfigure}[t]{0.15\textwidth}
			\includegraphics[keepaspectratio,width=\textwidth]{\idx}
			\caption{\name}
			\label{fig:cube\j}
		\end{subfigure}	
	\fi
}
\caption{Mud lake dataset used in the MTHS experiment with the corresponding acquisition dates. The area delineated in red in Fig.~\ref{fig:cube5} highlights a region known to contain outliers (this observation results from a previous analysis led on this dataset in \cite{Thouvenin2015b}).}
\label{fig:cube}
\end{figure*}

	\subsection{Results}

The results reported in Table~\ref{tab:results_synth} correspond to a single trial of the different algorithms. More precisely, the results reported for VCA/FCLS are representative of the results obtained over multiple runs, which have not been observed to vary significantly from one run to another. A similar observation has been made for multiple runs of the asynchronous algorithms (ASYNC and DAVIS) whose performance does not change significantly over different runs for the simulation setting adopted in this paper, both in terms of estimation accuracy and computation time.

\begin{itemize}
    \item \textbf{Endmember estimation:} the proposed asynchronous algorithm leads to competitive endmember estimation for the three synthetic datasets (in terms of aSAM and RE), notably in comparison with its synchronous counterpart. We can note that the DSPLR algorithm yields interesting estimation results for $\nendm = 3$, which however significantly degrade as $\nendm$ increases. This partly results from the matrix inversions involved in the update steps of~\cite{Tsinos2017}, which remain relatively sensitive to the conditioning of the involved matrices, and consequently to the choice of the regularization parameter of the augmented Lagrangian.
    \item \textbf{Abundance estimation:} the synchronous PALM algorithm leads to the best abundance estimation results, even in the absence of any additional regularization on the spatial distribution of the abundances. In this respect, we can note that the performance of PALM and its asynchronous version is relatively similar, and consistently outperforms the other unmixing methods.
    \item \textbf{Overall performance:} the performance measures reported in Table~\ref{tab:results_synth} show that the proposed distributed algorithm yields competitive estimation results, especially in terms of the required computational time when compared to its synchronous counterpart. To be more explicit, the evolution of the objective function versus the computation time shows the interest of the allowed asynchronicity to speed up the unmixing task, as illustrated in Fig.~\ref{fig:objective} (the computation time required by Algo.~\ref{alg:master} is almost 4 times lower than the one of its synchronous counterpart).
\end{itemize}

Note that even though the SYNC and ASYNC algorithms start from the same initial point, there is no guarantee that both methods converge to the same critical point, which essentially accounts for the differences in the results reported for both methods in Table~\ref{tab:results_synth}. For the asynchronous algorithms, another potential source of variability comes from the variations in the order the updates are performed from one run to another. For the simulation setting adopted in this work, such variations have not been observed to lead to significant differences in the estimation results.

\section{Experiments with real data} \label{sec:real_exp}

In practice, as emphasized earlier, distributed unmixing procedures are of particular interest when considering the unmixing of large HS images, or of a sequence of HS images acquired by possibly different sensors at different time instants~\cite{Henrot2016,Thouvenin2015b,Yokoya2017}, referred to as multitemporal HS (MTHS) images. The unmixing of two large real HS images is first proposed, whereas the application to MTHS images essentially motivates the last example addressed in this section. The experiments have been conducted in the same setting as in the previous section (the pixels composing the datasets are evenly distributed between $\Omega = 3$ workers).

	\subsection{Description of the datasets}

        \paragraph{Cuprite dataset (single HS image)}
the first dataset considered in this work consists of a $190 \times 250$ subset extracted from the popular Cuprite dataset. In this case, reference abundance maps are available from the literature (see for instance~\cite{Nascimento2005,Miao2007}). After removing water-absorption and low SNR bands, $189$ out of the $224$ spectral bands initially available were exploited in the subsequent unmixing procedure. The data have been unmixed with $\nendm = 10$ endmembers based on prior studies conducted on this dataset~\cite{Nascimento2005,Miao2007}.

        \paragraph{Houston dataset (single HS image)}
the second dataset considered hereafter was acquired over the campus of the University of Houston, Texas, USA, in June 2012~\cite{Debes2014}. The $152 \times 108$ scene of interest is composed of $144$ bands acquired in the wavelength range \SIrange{380}{1050}{\nano\metre}. The data have been unmixed with $\nendm = 4$ endmembers based on prior studies conducted on this dataset~\cite{Drumetz2016}.

        \paragraph{Mud lake dataset (MTHS images)}
we finally consider a real sequence of AVIRIS HS images acquired between 2014 and 2015 over the Mud Lake, located in the Lake Tahoe region (California, United States of America)\footnote{The images from which the interest of interest is extracted are freely available from the online AVIRIS flight locator tool at \url{http://aviris.jpl.nasa.gov/alt_locator/}.}. The $100 \times 100$ scene of interest is in part composed of a lake and a nearby field displayed in Fig.~\ref{fig:cube}. The images have been unmixed with $\nendm = 3$ endmembers based on results obtained from prior studies conducted on these data~\cite{Thouvenin2015b,Thouvenin2015TR}, and confirmed by the results of the noise-whitened eigengap algorithm (NWEGA) \cite{Halimi2016} reported in Table~\ref{tab:tab_ega}. After removing the water absorption bands, 173 out of the 224 available spectral bands were finally exploited.

	\subsection{Results}
	
	\begin{table}[h]
\caption{Endmember number $\nendm$ estimated by NWEGA \cite{Halimi2016} on each image of the Mud lake dataset.} 
	\begin{center}
	\resizebox{0.48\textwidth}{!}{%
		\begin{tabular}{@{}lcccccc@{}} \toprule
			& 04/10/2014 & 06/02/2014 & 09/19/2014 & 11/17/2014 & 04/29/2015 & 10/13/2015  \\ \cmidrule{2-7}
NWEGA       & 3 & 3 & 3 & 4 & 3 & 4 \\ \bottomrule	
		\end{tabular}
	}
	\end{center}
\label{tab:tab_ega} \vspace{-0.3cm}
\end{table}

Given the absence of ground truth for the different datasets (except the indications available in the literature for the Cuprite scene~\cite{Nascimento2005,Miao2007}), the estimation results obtained by the proposed algorithms are compared to the other unmixing procedures in terms of the RE and the aSAM introduced in~\eqref{eq:aSAM_Y} and~\eqref{eq:RE} respectively (see Table~\ref{tab:results_real}). The consistency of the estimated abundance maps, reported in Figs.~\ref{fig:A_cuprite} to~\ref{fig:A3_real}, is also considered when analyzing the different results.

	   \paragraph{Cuprite dataset}
except for the DSPLR algorithm, whose scale indeterminacy leads to results somewhat harder to interpret for this dataset, the results obtained by the different methods are relatively similar, both in terms of the estimated abundance maps and the recovered endmembers (see Fig.~\ref{fig:A_cuprite}).

	   \paragraph{Houston dataset}
the distributed algorithms yield abundance maps in agreement with the VCA/FCLS and SISAL algorithms (see Fig.~\ref{fig:A_houston}). We can note that the algorithms SYNC, ASYNC and DSPLR provide a more contrasted abundance map for the concrete than VCA/FCLS, SISAL/FCLS and DAVIS.

	   \paragraph{Mud lake dataset}
the algorithms SYNC, DAVIS~\cite{Davis2016} and ASYNC lead to particularly convincing abundance maps, in the sense that the abundances of the different materials (containing soil, water and vegetation) are consistently estimated for each time instant, contrary to VCA/FCLS, SISAL/FCLS and DSPLR (see Figs.~\ref{fig:A1_real} to~\ref{fig:A3_real}). At $\iworker = 5$, VCA/FCLS and SISAL, which have been applied  individually to each image of the sequence, appear to be particularly sensitive to the presence of outliers in the area delineated in red in Fig.~\ref{fig:cube5} (see~\cite{Thouvenin2015b} for a previous study on this dataset). This observation is further confirmed by the abundance maps reported at $t = 5$ in Figs.~\ref{fig:A1_real} and~\ref{fig:A2_real}, as well as the corresponding endmembers reported in Fig.~\ref{fig:real_endm} (whose amplitude is significantly greater than 1). This sensitivity notably results from the fact that each scene has been analyzed independently from the others in this specific context (note that the results would have been worse if these methods were applied to all the images at once).

	   \paragraph{Global reconstruction performance}
the performance measures reported for the different datasets in Table~\ref{tab:results_real} confirm the interest of the PALM algorithm and its asynchronous variant for unmixing applications. The asynchronous variant can be observed to lead to a notable reduction of the computation time (see also Fig.~\ref{fig:objective_real}), while allowing a reconstruction performance similar to the classical PALM algorithm to be obtained.

\setlength\columnsep{0.1pt} 
\begin{table}[!t] 
\vspace{-0.3cm}
\caption{Simulation results on real data (RE $\times 10^{-4}$).}
	\begin{center}
		\begin{tabular}{@{}lllccc@{}} \toprule
&	Algorithm &  	   & RE & aSAM($\Y$) (\textdegree) & time (\si{\second}) \\ \cmidrule{1-6}
\multirow{6}{*}{\rotatebox{90}{Cuprite}} 
&VCA/FCLS &\cite{Nascimento2005}     &  0.51 & 0.96 & \textbf{2} \\
&SISAL/FCLS	&\cite{Bioucas2009}   &  0.47 & 0.92 & 6 	\\
&DSPLR &\cite{Tsinos2017} &	1.25 &	1.42 &	20.2 \\   
&DAVIS &\cite{Davis2016} &	0.33 &	0.79 &	64.0 \\   
&SYNC &	&\textbf{0.15} &	\textbf{0.55} &	1290 \\
&ASYNC &&	0.30 &	0.77 &	134 \\  
 \cmidrule{1-6}
\multirow{6}{*}{\rotatebox{90}{Houston}} 
&VCA/FCLS &\cite{Nascimento2005}    &  22.5 & 3.31 & \textbf{0.1} \\
&SISAL/FCLS	&\cite{Bioucas2009}   &  21.3 & 2.01 & 0.6 	\\
&DSPLR &\cite{Tsinos2017} &	\textbf{0.13} &	\textbf{0.99} &	51.5 \\  
&DAVIS &\cite{Davis2016} &	14.9 &	2.44 &	22.3 \\  
&SYNC &	&0.21 &	1.14 &	84.6 \\  
&ASYNC &	&0.24 &	1.17 &	24.9 \\ 
\cmidrule{1-6}
\multirow{6}{*}{\rotatebox{90}{Mud lake}} 
&VCA/FCLS  &\cite{Nascimento2005}   &  23.7 & 13.23 & \textbf{1} \\
&SISAL/FCLS	&\cite{Bioucas2009}   &  \textbf{1.65} & \textbf{3.09} & 2 	\\ 
&DSPLR &\cite{Tsinos2017} &	1.93 &	10.9 &	99.6 \\   
&DAVIS &\cite{Davis2016} &	17.61 &	6.27 &	58.9 \\  
&SYNC &	&5.05 &	5.88 &	70.4 \\ 
&ASYNC &	&5.13 &	5.88 &	35.0 \\  
\bottomrule 
		\end{tabular}
	\end{center}
\label{tab:results_real} \vspace{-0.3cm}
\end{table}
\section{Conclusion and future work} \label{sec:conclusion}

This paper focused on a partially asynchronous distributed unmixing algorithm based on recent contributions in non-convex optimization~\cite{Cannelli2016,Chang2016,Davis2016}, which proves convenient to address large scale hyperspectral unmixing problems. Under relatively standard conditions, the proposed approach inherits from the convergence guarantees studied in~\cite{Cannelli2016}, and from those of the traditional PALM algorithm~\cite{Bolte2013,Chouzenoux2016} for its synchronous counterpart. Evaluated on synthetic and real data, the proposed approach provided competitive estimation results, while significantly reducing the computation time to reach convergence. From a computational point of view, implementing a fully functional, large scale asynchronous unmixing algorithm and assessing its scalability with respect to the volume of data involved is an interesting prospect. As with any distributed algorithm, the computation time required by the proposed method
is expected to decrease linearly with the number of workers assigned to the unmixing task until the cost of the master/worker communications is comparable to the cost of the estimation task conducted on each worker. Future research perspectives also include the extension to different network topology as in~\cite{Pesquet2014,Bianchi2016}, or the use of variable metrics as described in~\cite{Repetti2014,Chouzenoux2014,Chouzenoux2016,Frankel2015}.

\begin{appendix}
\section{Convergence proof} \label{sec:cv_proof}
The proposed sketch of proof adapts the first arguments developed in~\cite{Cannelli2016}, in order to clarify that the proposed algorithm fits within this general framework. Note that a similar proof can be obtained by induction when $J$ blocks have to be updated by each worker, and $I$ blocks by the master node (corresponding to the situation described in~\eqref{eq:problem}).

\begin{lemma} \label{lemma1}
Under Assumptions~\ref{alg_assumption} to~\ref{assumption2}, there exists two positive constants $c_\x$ and $c_\z$ such that
\begin{align}
& \Psi(\x^{k+1},\z^{k+1}) \leq \Psi(\x^k,\z^k) \nonumber\\
& - \frac{\gamma_k}{2} \bigl( c_\x - \gamma_k (\Lx^+ + \Ldxz^+) \bigr) \norm{\xhtk^k - \xtk^k}^2  \nonumber\\
& - \frac{\gamma_k}{2} \bigl(c_\z - \gamma_k \Lz^+ \bigr) \norm{\zh^k - \z^k}^2 \nonumber\\
& + \frac{1}{2} \tau \Ldxz^+ \sum_{q = k-\tau+1}^k \norm{\z^q - \z^{q - 1}}^2 \label{eq:lemma1}.
\end{align}
\end{lemma}

\begin{proof}
\textbf{Step 1:} Assumption~\ref{assumption}\ref{assumption:partial_grad} allows the descent lemma \cite[p. 683]{Bertsekas1999} to be applied to $\z \mapsto F(\x, \z)$, leading to
\begin{align} \label{eq:descent_z}
F(\x^{k+1}, &\z^{k+1}) \leq F(\x^{k+1},\z^k) + \frac{\Lz^k}{2} \norm{\z^{k+1} - \z^k}^2 \nonumber \\
&+ \pscalar{\nabla_{\z} F(\x^{k+1},\z^k), \z^{k+1} - \z^k} 
\end{align}
Thus,
\begin{align}
    &\Psi(\x^{k+1},\z^{k+1}) \leq F(\x^{k+1},\z^k) + G(\x^{k+1}) \nonumber \\
    & \qquad + \pscalar{\nabla_{\z} F(\x^{k+1},\z^k), \z^{k+1} - \z^k}  \nonumber \\
    & \qquad + \frac{\Lz^k}{2} \norm{\z^{k+1} - \z^k}^2 + r(\z^{k+1}) \nonumber
\end{align}
    
\begin{align}
    & = \ftk(\xtk^{k+1},\z^{k}) + \gtk(\xtk^{k+1}) + \frac{\Lz^k}{2} \norm{\z^{k+1} - \z^k}^2 \nonumber \\
    & \quad + \sum_{q \neq \omega^k} f_q(\x_q^{k},\z^k) + g_q(\x_q^{k}) + r(\z^{k+1})  \nonumber \\
    & \quad + \pscalar{\nabla_{\z} F(\x^{k+1},\z^k), \z^{k+1} - \z^k}.
\end{align}
Since $\z^{k+1} = \z^k + \gamma^k \bigl( \zh^k - \z^k  \bigr)$, we further have
\begin{align} \label{eq:inequality}
\Psi(\x^{k+1},& \z^{k+1}) \leq \ftk(\xtk^{k+1},\z^{k}) + \gtk(\xtk^{k+1}) \nonumber \\
& \quad + \sum_{q \neq \omega^k} f_q(\x_q^{k},\z^k) + g_q(\x_q^{k}) + r(\z^{k+1})  \nonumber \\
& \quad + \pscalar{\nabla_{\z} F(\x^{k+1},\z^k), \z^{k+1} - \z^k} \nonumber \\
& \quad + \frac{(\gamma^k)^2\Lz^k}{2} \norm{\zh^{k} - \z^k}^2.
\end{align}
In addition, the optimality of $\zh^k$ implies
\begin{equation} \label{eq:optimality_z}
\begin{split}
r(\zh^k) &+ \frac{c_{\z}^k}{2} \norm{\zh^k - \z^k}^2 \\
& + \pscalar{\nabla_{\z} F(\x^{k+1},\z^k), \zh^k - \z^k}  \leq r(\z^k)
\end{split}
\end{equation}
and the convexity of $r$ leads to
\begin{equation} \label{eq:convexity_r}
r(\z^{k+1}) \leq r(\z^k) + \gamma_k \bigl( r(\zh^k) -  r(\z^k) \bigr).
\end{equation}
Combining~\eqref{eq:convexity_r},~\eqref{eq:optimality_z} and exploiting the expression $\z^{k+1} = \z^k + \gamma^k \bigl( \zh^k - \z^k  \bigr)$ leads to
\begin{align}
    r(\z^{k+1}) &\leq r(\z^k) + \gamma_k \bigl( r(\zh^k) -  r(\z^k) \bigr) \nonumber \\
    \text{(from \eqref{eq:optimality_z})} & \leq r(\z^k) - \frac{\gamma^k c_{\z}^k}{2} \norm{\zh^{k} - \z^k}^2 \nonumber \\
    & - \gamma^k \pscalar{\nabla_{\z} F(\x^{k+1},\z^k), \zh^k - \z^k}. \label{eq:stepz}
\end{align}
Combining~\eqref{eq:stepz} and~\eqref{eq:inequality} finally results in
\begin{align} \label{eq:sufficient_decrease_z}
    &\Psi(\x^{k+1},\z^{k+1}) \leq \ftk(\xtk^{k+1},\z^{k}) + \gtk(\xtk^{k+1}) \nonumber \\
    & \quad + r(\z^{k}) + \sum_{q \neq \omega^k} f_q(\x_q^{k},\z^k) + g_q(\x_q^{k}) \nonumber \\
    & \quad - \frac{\gamma_k}{2} (c_{\z}^k - \gamma^k \Lz^k) \norm{\zh^{k} - \z^k}^2.
\end{align}

\textbf{Step 2:} Arguments similar to those used in Step 1 above lead to
\begin{align} \label{eq:step_x}
    & \ftk (\xtk^{k+1}, \z^k) + \gtk (\xtk^{k+1}) \leq \ftk (\xtk^k, \z^k) & \nonumber \\
    & + \bigl\langle \nabla_{\xtk} \ftk (\xtk^k,\z^k) - \nabla_{\xtk} \ftk (\xtk^k,\zh^k), \xtk^{k+1} - \xtk^k \bigr\rangle & \nonumber \\
    & - \frac{\gamma_k}{2} \bigl( c_{\xtk}^k - \gamma_k \Lxtk^k \bigr) \norm{\xhtk^k - \xtk^k}^2 + \gtk(\xtk^k). &
\end{align}
Since $\nabla_{\x_\omega} f_\omega (\x_\omega,\cdot)$ is assumed to be Lipschitz continuous (see Assumption~\ref{assumption2}\ref{assumption2_lip}), we have
\begin{equation*}
\begin{split}
&\bigl\langle \nabla_{\xtk} \ftk (\xtk^k,\z^k) - \nabla_{\xtk} \ftk (\xtk^k,\zh^k), \xtk^{k+1} - \xtk^k \bigr\rangle \\
& \quad \leq  \Ldxz^k \norm{\z^k - \zh^k} \norm{\xtk^{k+1} - \xtk^k}
\end{split}
\end{equation*}
which, combined with~\eqref{eq:step_x}, leads to
\begin{equation} \label{eq:sufficient_decrease_x}
\begin{split}
& \ftk (\xtk^{k+1}, \z^k) + \gtk (\xtk^{k+1}) \leq \ftk (\xtk^k, \z^k) \\
& \quad + \Ldxz^k \norm{\z^k - \zh^k} \norm{\xtk^{k+1} - \xtk^k} + \gtk(\xtk^k)  \\
& \quad - \frac{\gamma_k}{2} \bigl( c_{\xtk}^k - \gamma_k \Lxtk^k \bigr) \norm{\xhtk^k - \xtk^k}^2 .
\end{split}
\end{equation}

\textbf{Step 3:} From this point, the product involving $\norm{\z^k - \tilde{\z}^k}$ in~\eqref{eq:sufficient_decrease_x} can be bounded as proposed in \cite[Theorem 5.1]{Davis2016}. To this end, we first note that
\begin{align} \label{eq:step1}
    & \Ldxz^k \norm{\z^k - \zh^k} \norm{\xtk^{k+1} - \xtk^k} \nonumber \\
    & \leq \frac{\Ldxz^k}{2} \norm{\z^k - \zh^k}^2 + \frac{\Ldxz^k}{2} \norm{\xtk^{k+1} - \xtk^k}^2 \nonumber \\
    & = \frac{\Ldxz^k}{2} \norm{\z^k - \zh^k}^2 + \frac{\Ldxz^k \gamma_k^2}{2} \norm{\xhtk^k - \xtk^k}^2  \nonumber \\
    & \text{(using $\xtk^{k+1} = \xtk^{k} + \gamma_k (\xhtk^{k} - \xtk^{k})$)}. 
\end{align}
Besides, using the fact that $\dtk \leq \tau$ for any index $k$ (see Assumption~\ref{alg_assumption}), we have
\begin{align} \label{eq:step2}
&\norm{\z^k - \tilde{\z}^k}^2 = \norm{\sum_{q = k-\dtk+1}^k (\z^q - \z^{q - 1})}^2 \nonumber \\
& \leq \tau \sum_{q = k-\tau+1}^k \norm{\z^q - \z^{q - 1}}^2.
\end{align}
Combining~\eqref{eq:sufficient_decrease_x}, \eqref{eq:step1}, and~\eqref{eq:step2} then leads to
\begin{align} \label{eq:step3}
    \begin{split}
    & \ftk (\xtk^{k+1}, \z^k) + \gtk (\xtk^{k+1}) \leq \ftk (\xtk^k, \z^k) \\
    & - \frac{\gamma_k}{2} \bigl( c_{\xtk}^k - \gamma_k (\Lxtk^k + \Ldxz^k )  \bigr) \norm{\xhtk^k - \xtk^k}^2 \\
    & + \tau \Ldxz^k \sum_{q = k-\tau+1}^k \norm{\z^q - \z^{q - 1}}^2 + \gtk(\xtk^k) .
    \end{split}
\end{align}
\textbf{Step 4:} Combining~\eqref{eq:sufficient_decrease_z}, \eqref{eq:step3} and using the bounds on the different Lipschitz constants introduced in Assumptions~\ref{assumption}\ref{assumption_lip} and~\ref{assumption2}\ref{assumption2_lip} finally leads to the announced result.
\end{proof}

According to Lemma~\ref{lemma1}, the objective function $\Psi$ is not necessarily decreasing from an iteration to another due to the presence of a residual term involving $\tau$ past estimates of $\z$. From this observation, an auxiliary function (whose derivation is reproduced in Lemma~\ref{lemma2} for the sake of completeness) has been proposed in~\cite{Davis2016}. The introduction of such a function, which is eventually non-increasing between two consecutive iterations, is of particular interest for the convergence analysis. This function finally allows convergence guarantees related to the original problem~\eqref{eq:problem} to be recovered.

\begin{lemma}[Auxiliary function definition, adapted from \protect{\cite[Proof of Theorem 5.1]{Davis2016}}] \label{lemma2}
Under the same assumptions as in Lemma~\ref{lemma1}, let $\Phi$ be the function defined by
\begin{align}
&\Phi \bigl( \x(0),\z(0),\z(1),\dotsc,\z(\tau) \bigr) = \Psi\bigl( \x(0),\z(0) \bigr) \nonumber \\
& + \frac{\beta}{2} \sum_{q = 1}^\tau (\tau - q + 1) \norm{\z(q) - \z(q-1)}^2
\end{align}
with $\beta = \tau \Ldxz^+$.
Let $\w^k = (\x^k,\z^k,\check{\z}^k)$ and $\check{\z}^k = (\z^{k-1},\dotsc,\z^{k-\tau})$ for any iteration index $k \in \mathbb{N}$ (with the convention $\z^{q} = \z^0$ if $q < 0$). Then,
\begin{align} \label{eq:eq_lemma2}
& \Phi(\w^{k+1}) \leq \Phi(\w^k) \nonumber \\
& - \frac{\gamma_k}{2} \bigl( c_{\x} - \gamma_k (\Lx^+ + \Ldxz^+) \bigr) \norm{\xhtk^k - \xtk^k}^2  \nonumber\\
& - \frac{\gamma_k}{2} \bigl(c_{\z} - \gamma_k (\Lz^+ + \tau^2 \Ldxz^+ ) \bigr) \norm{\zh^k - \z^k}^2 .
\end{align}
\end{lemma}
\begin{proof}
The expression of the auxiliary function proposed in~\cite{Davis2016} results from the following decomposition of the residual term $\sum_{q = k - \tau + 1}^k \norm{\z^q - \z^{q-1}}^2$. Introducing the auxiliary variables
\begin{equation*}
\alpha^k = \sum_{q = k-\tau+1}^k (q - k + \tau) \norm{\z^q - \z^{q - 1}}^2
\end{equation*}
we can note that
\begin{equation} \label{eq:diff_alpha}
\alpha^k - \alpha^{k+1} = \sum_{q = k-\tau+1}^k \norm{\z^q - \z^{q - 1}}^2 - \tau \norm{\z^{k+1} - \z^k}^2.
\end{equation}
Thus, using the upper bound $\Ldxz^k \leq \Ldxz^+$ (Assumption~\ref{assumption2}\ref{assumption2_lip}) and replacing~\eqref{eq:diff_alpha} in~\eqref{eq:lemma1} yields
\begin{align*}
& \Psi(\x^{k+1},\z^{k+1}) + \beta \alpha^{k+1}\leq \Psi(\x^k,\z^k) + \beta \alpha^k \\
& - \frac{\gamma_k}{2} \bigl( c_\x - \gamma_k (\Lx^+ + \Ldxz^+) \bigr) \norm{\xhtk^k - \xtk^k}^2  \nonumber\\
& - \frac{\gamma_k}{2} \bigl(c_{\z} - \gamma_k (\Lz^+ + \tau^2 \Ldxz^+ ) \bigr) \norm{\zh^k - \z^k}^2.
\end{align*}
Observing that $\Phi(\w^k) = \Psi(\x^{k},\z^{k}) + \alpha^{k}$ finally leads to the announced result.
\end{proof}

The previous lemma makes clear that the proposed algorithm can be studied as a special case of~\cite{Cannelli2016}. The rest of the convergence analysis, which involves somewhat convoluted arguments, exactly follows~\cite{Cannelli2016} up to minor notational modifications.

\end{appendix}

\begin{figure*}
\centering
\includegraphics[keepaspectratio,width=0.97\textwidth]{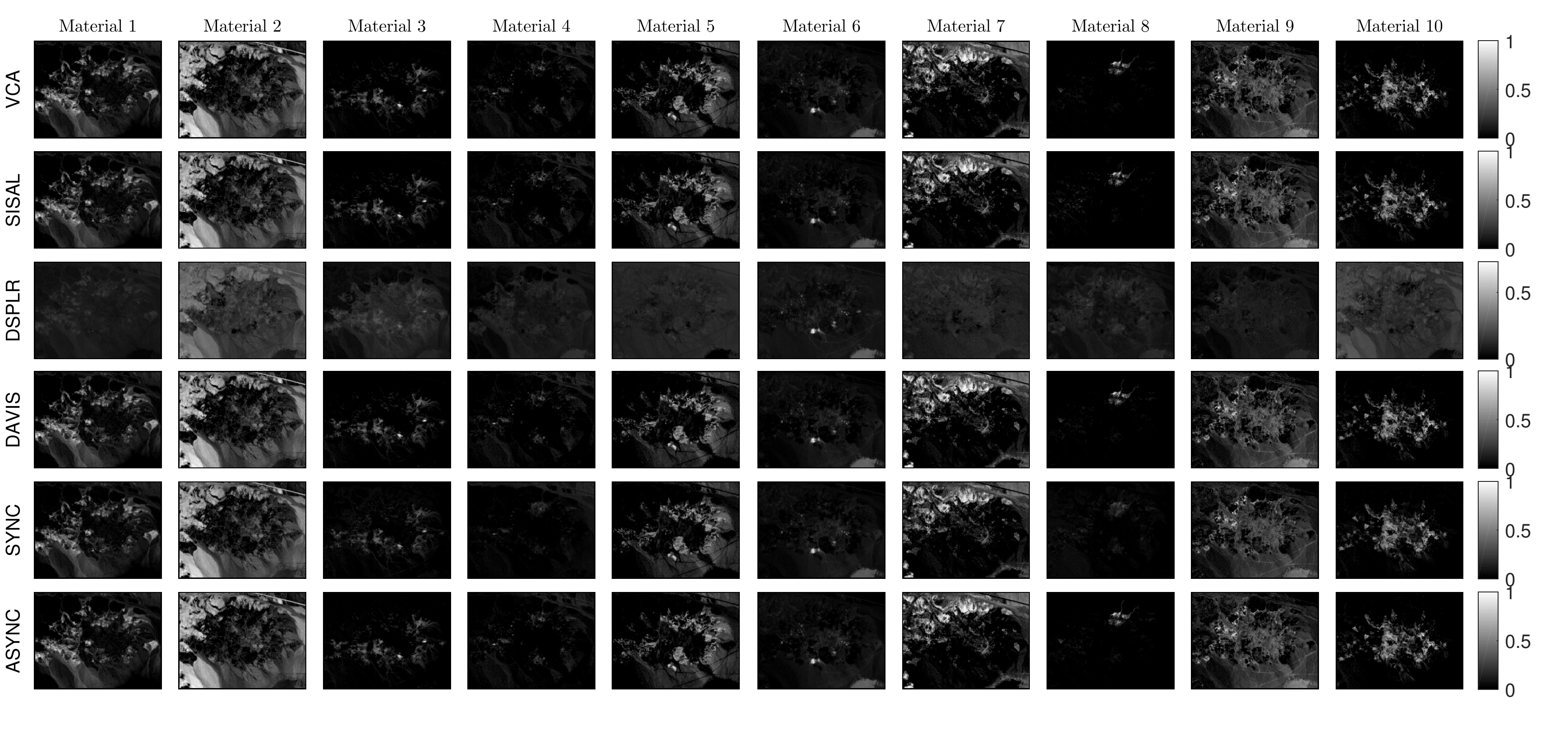}
\caption{Abundance maps recovered by the different methods (in each row) for the Cuprite dataset.}
\label{fig:A_cuprite}
\end{figure*}

\begin{figure}[t!]
\centering
\includegraphics[keepaspectratio,width=0.47\textwidth]{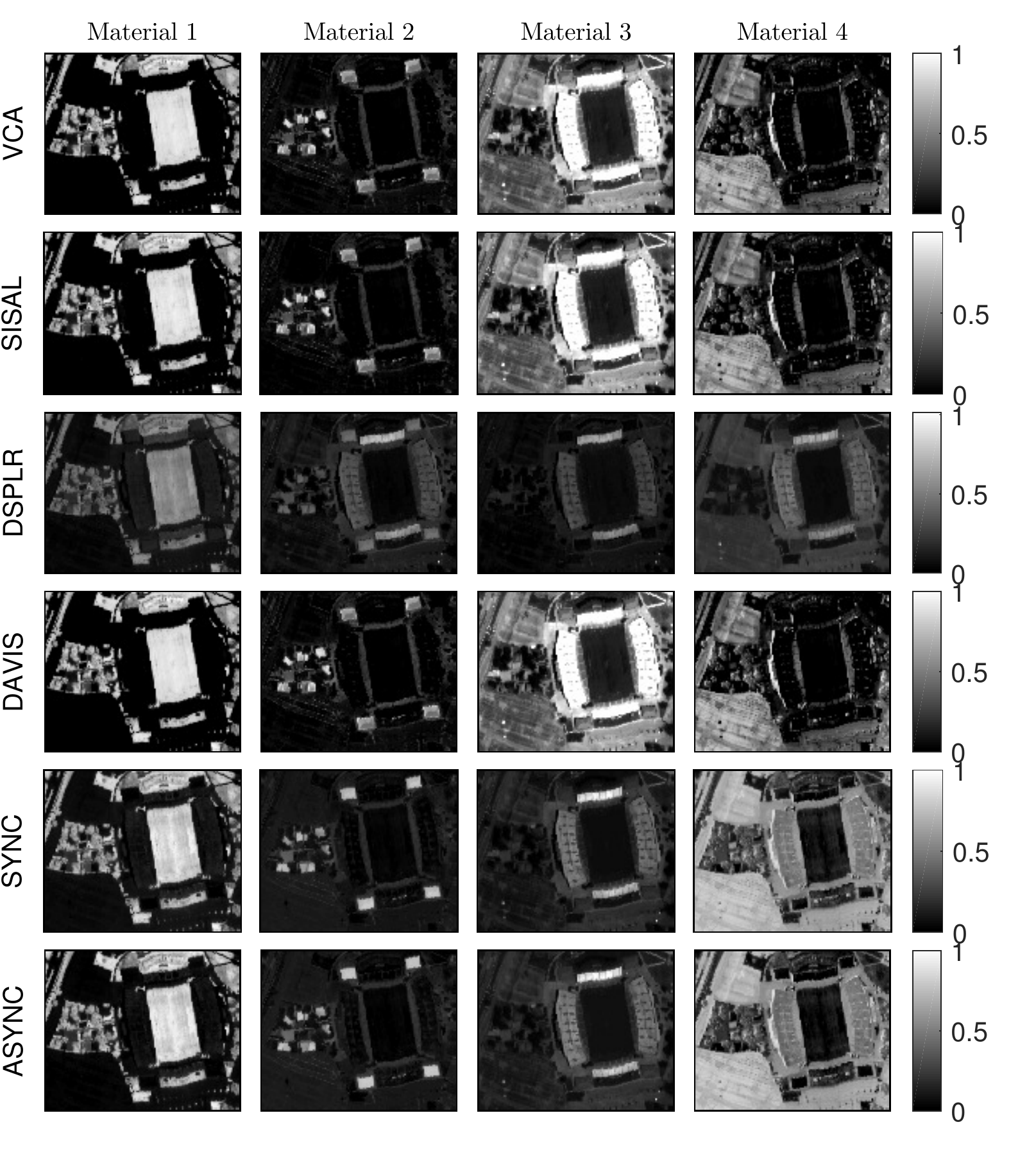}
\caption{Abundance maps recovered by the different methods (in each row) for the Houston dataset.}
\label{fig:A_houston}
\end{figure}

%
\def\materials{{Soil},{Water},{Vegetation}}

\foreach \mat[count=\j] in \materials {
\def\idx{\the\numexpr\j+11} 
\begin{figure}[t!]
\centering
\includegraphics[keepaspectratio,width=0.49\textwidth]{\idx}
\caption{\mat{} abundance map recovered by the different methods (in each row) at each time instant (given in column) for the experiment on the Mud lake dataset [the different rows correspond to VCA/FCLS, SISAL/FCLS, DSPLR \cite{Tsinos2017}, DAVIS \cite{Davis2016}, SYNC and ASYNC methods].}
\label{fig:A\j_real}
\end{figure}}
\def\names{{vca},{sisal},{dsplr},{davis},{sync},{async}}
\def\materials{{Soil},{Water},{Veg.}}

\begin{figure}[t]
\centering
\foreach \name [count=\i] in \names {
	\foreach \material[count=\j] in \materials {
		\def\idx{\the\numexpr\j+(\i-1)*3+14} 
		\begin{subfigure}[t]{0.15\textwidth}	
		\includegraphics[keepaspectratio,width=\textwidth]{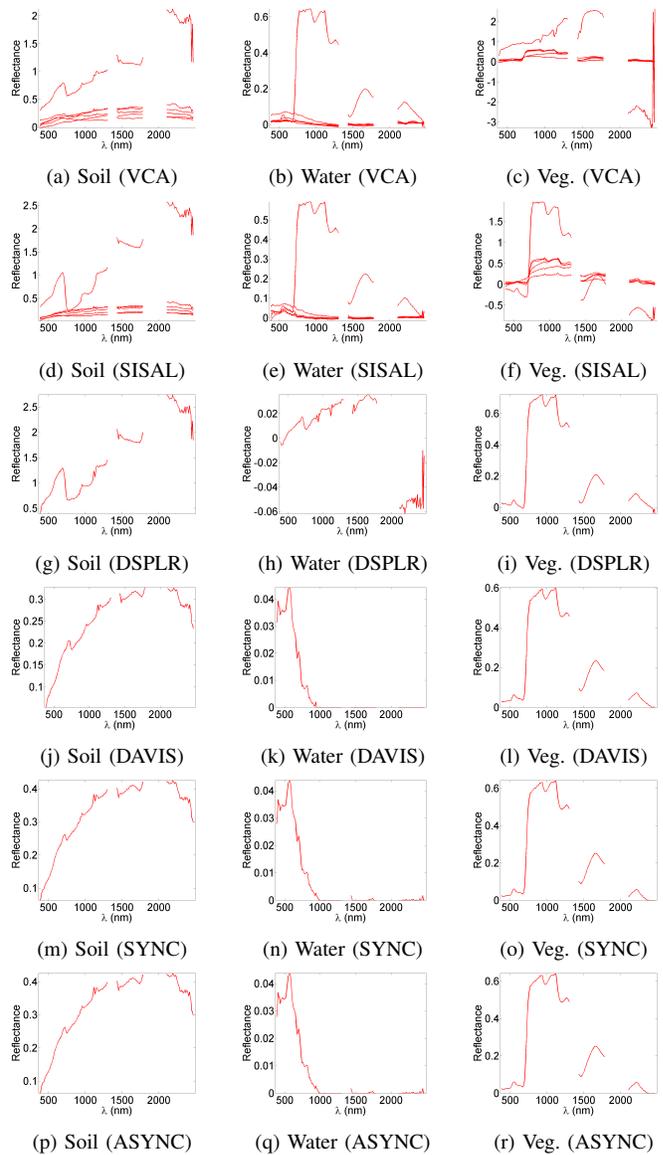}
		\caption{\material{} (\MakeUppercase{\name})}
		\label{fig:endm_\j_\name}
		\end{subfigure}
	} \\
}
\caption{Endmembers ($\m_r$, red lines) recovered by the different methods from the real dataset depicted in Fig.~\ref{fig:cube}. Endmembers extracted by VCA, SISAL and DSPLR show a notable sensitivity to the presence of outliers in these data, }
\label{fig:real_endm}
\end{figure}

\begin{figure*}[t!]
    \centering
    \begin{subfigure}[t]{0.32\textwidth}
    	\centering	\includegraphics[keepaspectratio,width=0.99\textwidth]{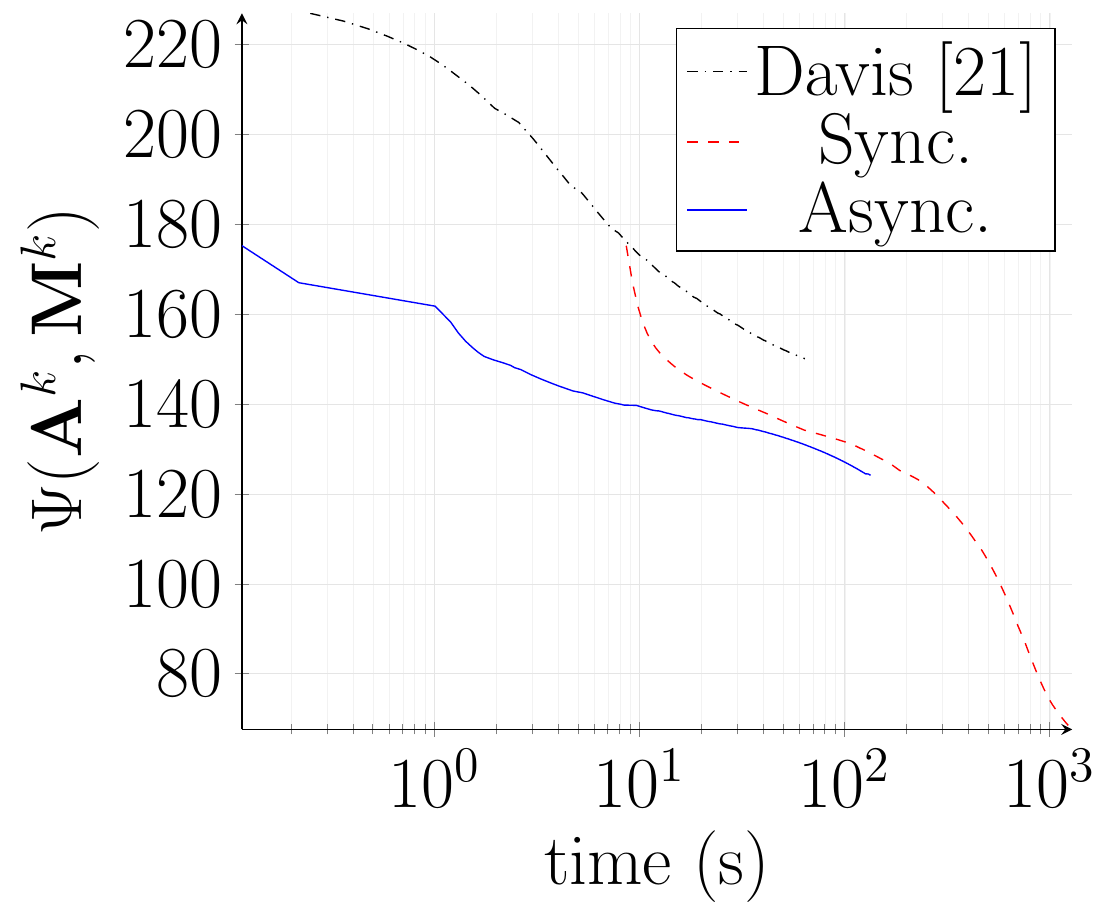}
    	\caption{Cuprite}
    	\label{fig:f_cuprite}
    \end{subfigure}
    \begin{subfigure}[t]{0.32\textwidth}
    	\centering	\includegraphics[keepaspectratio,width=0.99\textwidth]{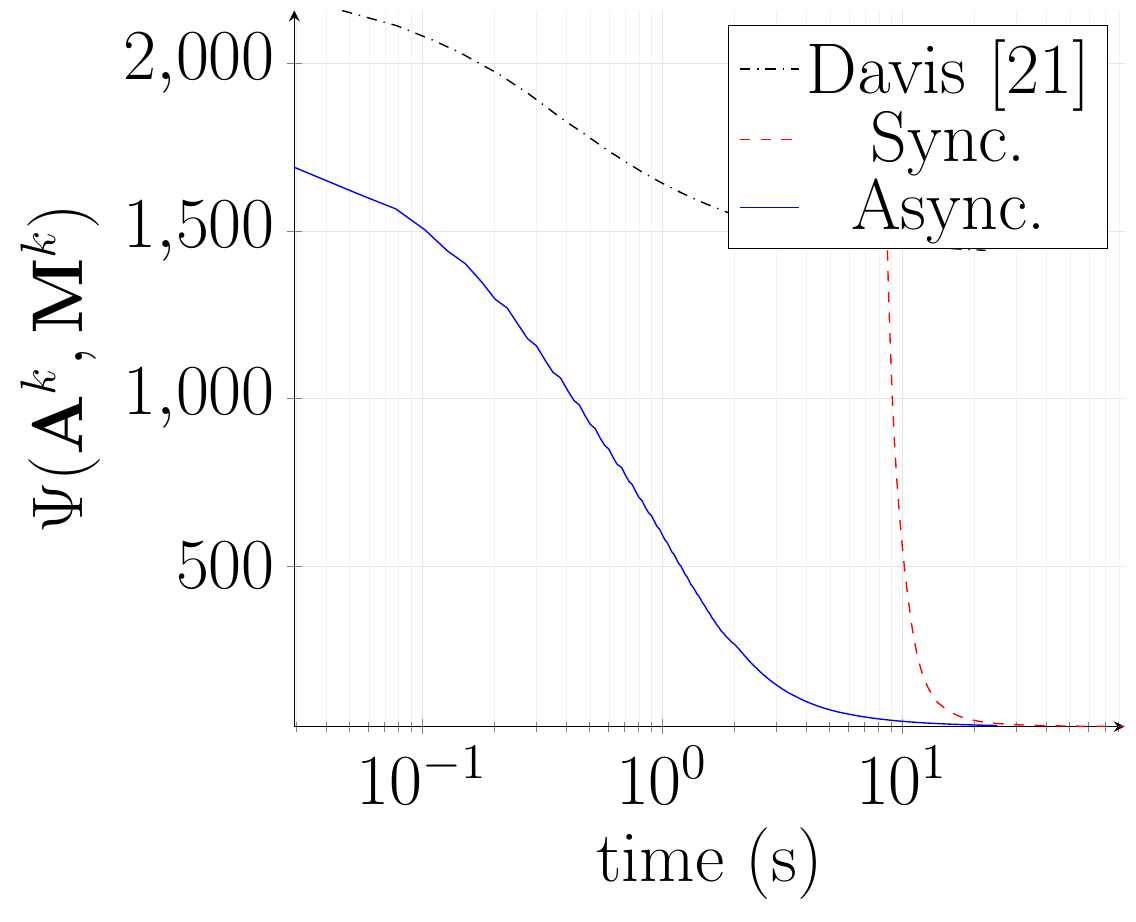}
    	\caption{Houston}
    	\label{fig:f_houston}
    \end{subfigure}
    \begin{subfigure}[t]{0.32\textwidth}
    	\centering	\includegraphics[keepaspectratio,width=0.99\textwidth]{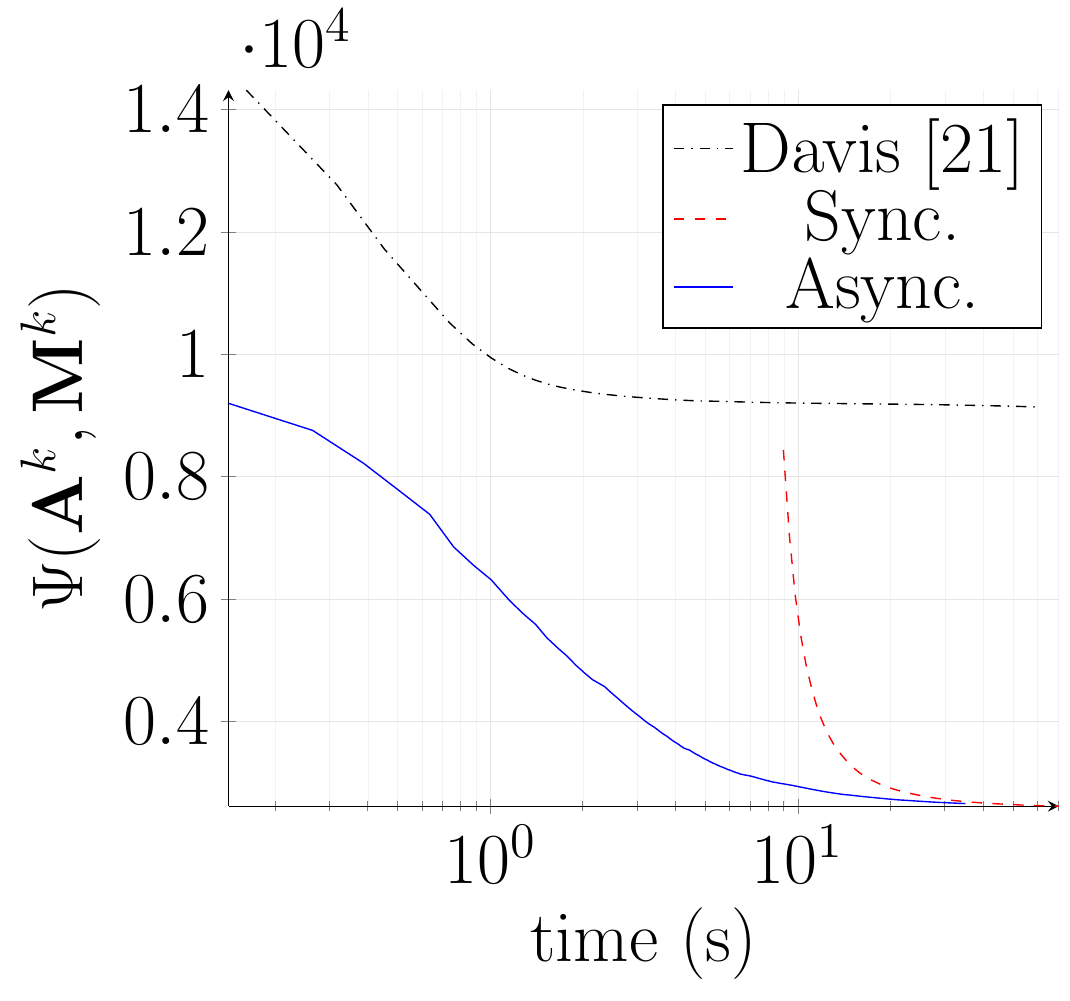}
    	\caption{Mud lake}
    	\label{fig:f_mud}
    \end{subfigure}
    \caption{Evolution of the objective function for the synthetic datasets, obtained for DAVIS~\cite{Davis2016}, Algo.~\ref{alg:master} and its synchronous version until convergence.}
    \label{fig:objective_real}
\end{figure*}


\bibliographystyle{IEEEtran}
\bibliography{strings_all_ref,all_ref}

\begin{IEEEbiography}[{\includegraphics[width=1in,height=3in,clip,keepaspectratio]{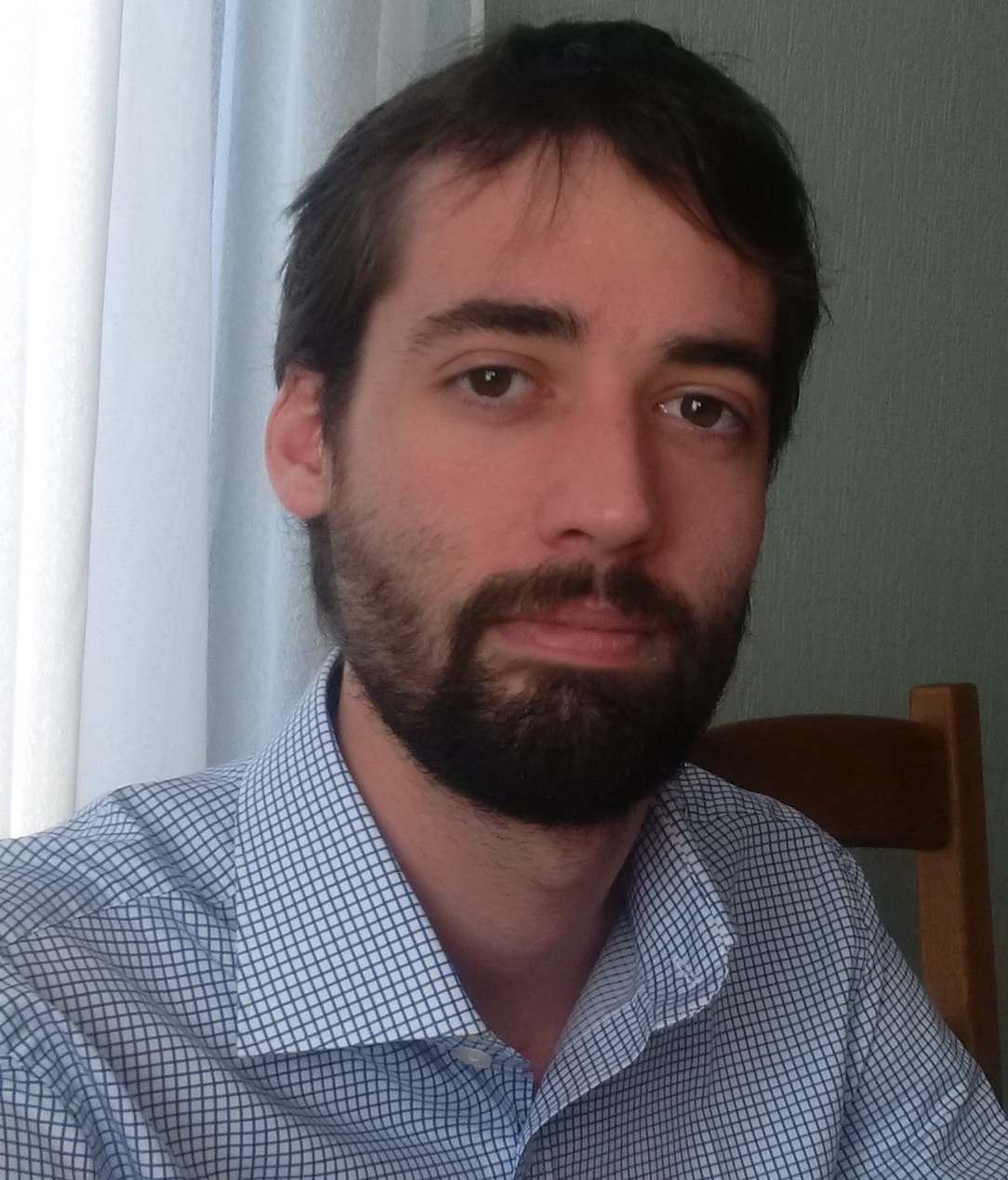}}]{Pierre-Antoine Thouvenin} (S'15--M'17) received the state engineering degree in electrical engineering from ENSEEIHT, Toulouse, France, and the M.Sc. degree in signal processing from the National Polytechnic Institute of Toulouse (INP Toulouse), both in 2014, as well as the PhD degree in Signal Processing from the INP Toulouse in 2017. Since September 2017, he has been working as a post-doctoral research associate within the Biomedical and Astronomical Signal Processing (BASP) group, Heriot-Watt University, Edinburgh, UK. His research interests include statistical modeling, optimization techniques and hyperspectral unmixing.
\end{IEEEbiography}

\begin{IEEEbiography}[{\includegraphics[width=1in,height=3in,clip,keepaspectratio]{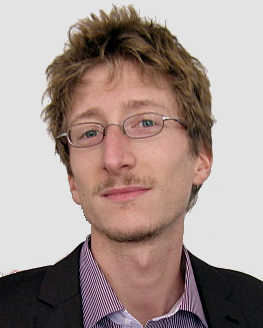}}]{Nicolas Dobigeon} (S'05--M'08--SM'13) received the state engineering degree in electrical engineering from ENSEEIHT, Toulouse, France, and the M.Sc. degree in signal
processing from the National Polytechnic Institute of Toulouse (INP Toulouse), both in June 2004, as well as the Ph.D. degree and Habilitation {\`a} Diriger des Recherches in Signal Processing from the INP Toulouse in 2007 and 2012, respectively.
He was a Post-Doctoral Research Associate with the Department of Electrical Engineering and Computer Science, University of Michigan, Ann Arbor, MI, USA, from 2007 to 2008.

Since 2008, he has been with the National Polytechnic Institute of Toulouse (INP-ENSEEIHT, University of Toulouse) where he is currently a Professor. He conducts his research within the Signal and Communications Group of the IRIT Laboratory and he is also an affiliated faculty member of the Telecommunications for Space and Aeronautics (T{\'e}SA) cooperative laboratory.
His current research interests include statistical signal and image processing, with a particular interest in Bayesian inverse problems with applications to remote sensing, biomedical imaging and genomics.
\end{IEEEbiography}

\begin{IEEEbiography}[{\includegraphics[width=1in,height=3in,clip,keepaspectratio]{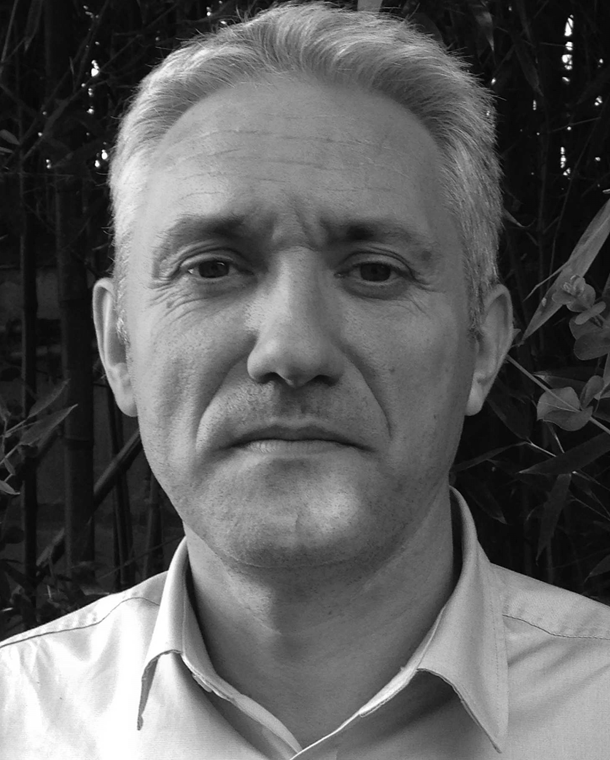}}]{Jean-Yves Tourneret} (SM'08) received the ing{\'e}nieur degree in electrical engineering from the Ecole Nationale Sup{\'e}rieure d'Electronique, d'Electrotechnique, d'Informatique, d'Hydraulique et des T{\'e}l{\'e}communications (ENSEEIHT) de Toulouse in 1989 and the Ph.D. degree from the National Polytechnic Institute from Toulouse in 1992. He is currently a professor in the university of Toulouse (ENSEEIHT) and a member of the IRIT laboratory (UMR 5505 of the CNRS). His research activities are centered around statistical signal and image processing with a particular interest to Bayesian and Markov chain Monte-Carlo (MCMC) methods. He has been involved in the organization of several conferences including the European conference on signal processing EUSIPCO'02 (program chair), the international conference ICASSP'06 (plenaries), the statistical signal processing workshop SSP'12 (international liaisons), the International Workshop on Computational Advances in Multi-Sensor Adaptive Processing CAMSAP 2013 (local arrangements), the statistical signal processing workshop SSP'2014 (special sessions), the workshop on machine learning for signal processing MLSP'2014 (special sessions). He has been the general chair of the CIMI workshop on optimization and statistics in image processing hold in Toulouse in 2013 (with F. Malgouyres and D. Kouam{\'e}) and of the International Workshop on Computational Advances in Multi-Sensor Adaptive Processing CAMSAP 2015 (with P. Djuric). He has been a member of different technical committees including the Signal Processing Theory and Methods (SPTM) committee of the IEEE Signal Processing Society (2001-2007, 2010-present). He has been serving as an associate editor for the IEEE Transactions on Signal Processing (2008-2011, 2015-present) and for the EURASIP journal on Signal Processing (2013-present).
\end{IEEEbiography}

\end{document}